\documentclass[letterpaper,twocolumn,10pt]{article}
\usepackage{usenix}

\usepackage{tikz}
\usepackage{amsmath}
\usepackage{color}

\usepackage{filecontents}
\usepackage{booktabs}
\usepackage{graphicx}
\usepackage{booktabs}
\usepackage{amssymb}
\usepackage{multirow}
\usepackage{algorithmicx}
\usepackage[vlined,ruled, linesnumbered]{algorithm2e}
\usepackage{subfigure}
\usepackage{amsmath}
\usepackage{amsthm} % add cuba at the end of proof
\usepackage{makecell}
\usepackage{enumitem}
\usepackage{pifont}
\newtheorem{theorem}{Theorem}

\newtheorem{remark}{Remark}
\setlist{nosep}
\SetKwIF{IfNot}{ElseIfNot}{}{if not}{then}{else if not}{}{}
\newcommand{\vect}[1]{\ensuremath{\mathbf{#1}}}
\newcommand{\x}{\vect{x}}
\newcommand{\y}{\vect{y}}

\begin{document}
%-------------------------------------------------------------------------------

%don't want date printed
\date{}

% make title bold and 14 pt font (Latex default is non-bold, 16 pt)
\title{United We Defend: Collaborative Membership Inference Defenses \\ in Federated Learning}

%for single author (just remove % characters)
\author{%
	\rm 
	Li Bai\textsuperscript{1}, Junxu Liu\textsuperscript{1,2}, Sen Zhang\textsuperscript{1}, Xinwei Zhang\textsuperscript{1}, Qingqing Ye\textsuperscript{1}, 
	Haibo Hu\textsuperscript{1,2} \thanks{Corresponding author: \texttt{haibo.hu@polyu.edu.hk}} \\
	\rm Department of Electrical and Electronic Engineering, The Hong Kong Polytechnic University\textsuperscript{1} \\
 	\rm PolyU Research Centre for Privacy and Security Technologies in Future Smart Systems\textsuperscript{2} \\
}

\maketitle
%-------------------------------------------------------------------------------
\begin{abstract}
%-------------------------------------------------------------------------------
Membership inference attacks (MIAs), which determine whether a specific data point was included in the training set of a target model, have posed severe threats in federated learning (FL).
Unfortunately, existing MIA defenses, typically applied independently to each client in FL, are ineffective against powerful trajectory-based MIAs that exploit temporal information throughout the training process to infer membership status.
In this paper, we investigate a new FL defense scenario driven by heterogeneous privacy needs and privacy-utility trade-offs, where only a subset of clients are defended, as well as a collaborative defense mode where clients cooperate to mitigate membership privacy leakage.
To this end, we introduce CoFedMID, a collaborative defense framework against MIAs in FL, which limits local model memorization of training samples and, through a defender coalition, enhances privacy protection and model utility.
Specifically, CoFedMID consists of three modules: a class-guided partition module for selective local training samples, a utility-aware compensation module to recycle contributive samples and prevent their overconfidence, and an aggregation-neutral perturbation module that injects noise for cancellation at the coalition level into client updates.
Extensive experiments on three datasets show that our defense framework significantly reduces the performance of seven MIAs while incurring only a small utility loss. These results are consistently verified across various defense settings.

\end{abstract}

\section{Introduction}
\label{sec:intro}
Federated learning (FL)~\cite{MMR+17} has emerged as a popular distributed machine learning paradigm in which multiple clients (e.g., mobile devices or institutions) collaboratively train a global model by exchanging local model updates (e.g., weights or gradients), rather than sharing their raw data.
% Despite the advantages of data localization in FL, research has shown that it faces significant privacy threats to the local sensitive information, with membership inference attacks (MIAs) being particularly prominent~\cite{NMS+19, MLS+19}.
Despite direct data leakage being mitigated, indirect privacy risks from potential adversaries (e.g., the honest-but-curious server or clients) remain a concern. One prominent threat is membership inference attacks (MIAs), which aim to determine whether a specific data sample was part of a client's local training dataset~\cite{NMS+19, MLS+19}. Compared to centralized settings, adversaries in FL  often have access to detailed model information and historical model snapshots throughout the training process. This enables powerful trajectory-based MIAs against both local and global models~\cite{GYB+22, CHE+24, ZGL+25}, thereby making effective defenses considerably more challenging~\cite{BLH+24, ZGL+25}.

% % MIAs aim to determine whether a specific data sample was used in the training process of a target model. 
% MIAs in FL aim to determine whether a specific data sample was part of a client's local training dataset, thereby posing a significant threat to the privacy of sensitive information~\cite{NMS+19, MLS+19}.
% % ~\cite{SRS+17, LGW+19} and serve as a basis for more complex attacks~\cite{CVC+20, XYY+22}.
% Adversaries in the FL setting, such as the central server or clients, often obtain detailed model information and historical model snapshots throughout training. 
% This access facilitates powerful trajectory-based MIAs against local and global models~\cite{GYB+22, CHE+24, ZGL+25}, 
% which leverage temporal information in model behavior between member samples (i.e., used during training) and non-member samples (i.e., unseen by the target model).
% Previous defenses, whether designed for FL or adapted from centralized learning, fall short in mitigating these attacks~\cite{BLH+24, ZGL+25} since membership privacy leakage occurs through the FL training process.
% % seriously compromising the membership privacy of clients’ local data. %~\cite{GYB+22, ZGL+25}.
Beyond the limitations imposed by those powerful attacks, current defense mechanisms, whether specifically designed for FL or adapted from centralized frameworks, typically enforce uniform and independent privacy safeguards across all clients. In practice, such one-size-fits-all strategies fail to reflect the diverse privacy needs of participants. Clients with low-risk or publicly available data may not require the same level of protection as those handling sensitive information~\cite{JZY+15, EHS+15, LJL+24}, especially when stronger defenses come at the cost of reduced utility~\cite{MNJ+22}. Furthermore, existing research indicates that implementing defenses in isolation may yield weaker privacy guarantees compared to those achievable through collaborative approaches. For example, secure multi-party computation~\cite{GO+98, GYB+22} demonstrates that collaborative computation can provide stronger privacy protections than operating independently. %These issues motivate us to explore approaches that better balance model utility and defense effectiveness by taking into account clients’ willingness to defend and collaborate.

In this paper, we focus on a novel defense setting that differs from previous works in two key aspects:
(1) \textit{Partial Scenario:}
only a subset of clients choose to implement defense mechanisms according to their actual privacy needs; and
% instead of enforcing a uniform defense across all clients, we consider a partial scenario, where only a subset of clients defends.
% It accounts for heterogeneous privacy requirements among clients~\cite{EHS+15, LJL+24} and reduces unnecessary perturbation to the FL process.
% \bl{However, the different training processes between undefended and defended clients introduce a new challenge: samples that were previously hard to attack may become vulnerable.}
%
(2) \textit{Collaborative Mode:} defended clients can organize into a group, referred to as a defender coalition, to collaboratively mitigate MIAs.
% beyond the independent mode, we explore client collaboration to jointly mitigate MIAs. 
Such collaboration is feasible as many real-world FL scenarios present conditions and opportunities that naturally enable it. 
For example, it can be formed by multiple devices owned by the same user in cross-device FL~\cite{KJM+16}, or several institutions under a common administration in cross-silo FL~\cite{LTS+20, SME+20}. Nevertheless, this new defense setting introduces an additional vulnerability arising from the training divergence between defended and undefended clients. As shown in Figure \ref{fig:hamploss}, member samples from defended clients (the \textit{member} line) can be easily distinguished from non-member samples from undefended clients (the \textit{non-member-IFL} line), even though they are indistinguishable from non-members that are entirely external to the FL system (the \textit{non-member-OFL} line).

% \bl{Bridge}
% Under such a setting, we point out another challenge for defense design, which arises from the divergent training processes between defended and undefended clients: partial defenses acorss clinets make that member samples may become difficult to distinguish from completely unseen non-member samples, yet remain easily distinguishable from non-member samples from other clients.
% Under this new defense setting, we point out another challenge arising from the training divergence between defended and undefended clients: member samples from defended clients may become indistinguishable from non-members completely unseen by the FL system, yet remain easily distinguishable from non-member samples from undefended clients.
% \end{itemize}
To effectively mitigate membership privacy leakage for defended clients, it is essential to minimize behavioral discrepancies between their member samples and both types of non-member samples. Given that these behavioral differences will become more pronounced as local models increasingly memorize individual samples~\cite{ZCB+17}, our core idea is to limit the exposure of the training dataset, thereby constraining the extent of memorization in local models and, through federated aggregation, in the global model throughout the training process.
% From one to N
To this end, we present \textbf{CoFedMID}, a \underline{Co}llaborative \underline{Fed}erated \underline{M}embership \underline{I}nference \underline{D}efense framework, which consists of three key modules as follows.
First, to enforce the memorization constraint, we introduce a class-guided partition module that enables selective local training samples but ensures that the defender coalition captures the full data distribution to preserve model utility.
Restricting the usage of training samples inevitably leads to utility degradation, especially as the coalition size increases.
We then introduce a utility-aware compensation module, which reintroduces informative samples according to their contribution to model performance.
Last but not least, to further enhance defense strength, CoFedMID incorporates an aggregation-neutral perturbation module to inject carefully designed random noise into client updates.
% analysis
It is worth noting that our framework employs a defender coalition to strengthen both privacy protection and model utility, without incurring any additional privacy cost beyond standard FL. Although this study focuses on a single coalition, our framework can be naturally extended to multiple ones.

% It is worth noting that our framework leverages a defender coalition to enhance protection and utility without extra privacy cost beyond standard FL, as it requires no additional private information sharing.
% Moreover, while designed for a single coalition in this work, it readily generalizes to multiple ones.

% Expr
% Considering emerging trajectory-based MIAs targeting global or local models~\cite{YSG+18, LJL+23, CHE+24, ZGL+25, LHL+24}, we evaluate our defense framework on CIFAR100~\cite{KAH+09}, CIFAR10~\cite{KAH+09}, and TinyImageNet~\cite{YLX+15} datasets.
% Experimental results show that our method outperforms existing defenses in both mitigation effectiveness and model utility, whether in pair-client coalitions or when half of the clients participate in the defense.
% Moreover, we systematically analyze the effects of different CoFedMID configurations and evaluate its robustness against potential adaptive attacks.

In summary, we make the following key contributions:
% \begin{itemize}[noitemsep, topsep=0pt]
\begin{itemize}[leftmargin=0.2in]
    \item We present the first study to consider partial and collaborative defense settings in FL, beyond previous uniform and independent privacy-preserving approaches. 

    \item We propose CoFedMID, an effective FL defense framework that mitigates local data memorization and enhances protection and utility through client collaboration.

    \item We design three modules to balance defense and utility: class-guided partition for selective training samples, utility-aware compensation to recycle contributive samples and prevent overconfidence, and aggregation-neutral perturbation to boost defense with noise injection.

    \item We evaluate CoFedMID against seven trajectory-based MIAs and six baseline defenses on three benchmark datasets, demonstrating its superior defense capability and high utility under various FL settings.
\end{itemize}

The paper is organized as follows. Section~\ref{sec:prob} and Section~\ref{sec:chal} provide the problem formulation and discuss the limitations of existing methods, respectively. Section~\ref{sec:met} presents our CoFedMID framework, followed by experimental results in Section~\ref{sec:exp}. Section~\ref{sec:dis} provides limitations and discussion, Section~\ref{sec:rel} reviews related work, and Section~\ref{sec:con} concludes the paper.
\section{Problem Formulation}
\label{sec:prob}

\subsection{Federated Learning}
Federated Learning (FL)~\cite{MMR+17} is a distributed learning paradigm where multiple clients (e.g., mobile devices or institutions) collaboratively train a global model without sharing raw data. In this paper, we consider a typical supervised FL setup with a central server and a set of $K$ clients, where each client $k\in [K]$ holds a private dataset $D_k$. The objective is to learn a globally shared model $f$ with parameters $\tilde{\theta}\in \mathbb{R}^P$ by solving the following empirical risk problem
{\small
\begin{align*}
&\tilde{\theta} = \min_{\theta\in\mathbb{R}^P}\left\{\mathcal{L}(\theta) \triangleq \sum_{k=1}^K w_k \mathcal{L}_k(\theta; D_k)\right\}, \\
&\text{where } \mathcal{L}_k(\theta; D_k) \triangleq \frac{1}{|D_k|} \sum_{(\x,\y)\in D_k} \ell(f_\theta(\x), \y).
\end{align*}
}%
Note that $\y\in[0,1]^N$ is a one-hot encoded label over a class space $\mathcal{S}$ of size $N$, where $\y_c=1$ indicates the ground truth class. 
$f_\theta(\x)\in\mathbb{R}^N$ refers to the predicted confidence probabilities across the $N$ classes and $\ell(\cdot, \cdot)$ denotes the loss function (e.g., cross-entropy).
The most fundamental approach for solving the above FL optimization problem is federated averaging (FedAvg) \cite{MMR+17}, where the parameters of the global model are updated through iterative training rounds. During each round $t\in [T]$, each client initializes its local parameters with the current global parameters $\tilde{\theta}^t$, performs local optimization on its respective dataset $D_k$, and sends the updated local parameters $\theta_k^t$ to the server. The server then aggregates the received parameters using a weighted average, i.e., $\tilde{\theta}^{t+1} = \sum_{k=1}^K w_k \cdot\theta_k^t$.
% \begin{align*}
%  \tilde{\theta} \leftarrow \sum_{k\in[K]} w_k \cdot\theta_k,   
% \end{align*}
% where the notation $|\cdot|$ denotes the dataset size.

\subsection{Attack Formulation}

\noindent\textbf{Attack Objective.} 
Similar to how it works in a centralized setting, MIA in FL aims to determine whether a given sample $(\x,\y)$ was included in the target client's local training dataset $D_{\text{tar}}$. 
Formally, this can be formalized via a membership inference function:
\begin{align*}
\mathcal{A}(\theta, (\x,\y)) = 
\begin{cases}
1, & \text{if } (\x,\y) \in D_{\text{tar}} \quad \text{(member)} \\
0, & otherwise \quad \text{(non-member)}
\end{cases}
\end{align*}
where $\mathcal{A}$ denotes an attack model that takes as input a sample $(\x,\y)$ and the model parameters $\theta$ (e.g., local or global), and outputs a binary prediction indicating whether $(\x,\y)$ is a member of $D_{\text{tar}}$.

\noindent\textbf{Attacker Capability.}
Following prior work~\cite{GYB+22, CHE+24, LHL+24, ZGL+25}, we assume that adversaries may be either recipients of the global model parameters $\tilde{\theta}^t$ (i.e., the other clients) or recipients of local model parameters from all participating clients $\{\theta^t_k\}_{k\in [K]}$ (i.e., the central server). Specifically, these adversaries are considered \textit{honest-but-curious}, i.e., they adhere to the FL protocol but attempt to infer whether a particular sample is included in the \textit{target client}'s local dataset. 
We refer to attacks that exploit the temporal evolution of these model snapshots to infer membership information as \textbf{trajectory-based MIAs}.

\subsection{Defense Formulation}
\label{sec:defense_form}

\textbf{{Uniform and Partial Scenarios.}}
Without loss of generality, we design our framework from the client’s perspective.
Based on whether all clients participate in the defense, existing defense scenarios can be categorized as either \textit{uniform} or \textit{partial}. In uniform defenses, all clients are involved, whereas in partial defenses, only a subset of clients participates.
The uniform defense scenario has been extensively explored in existing works~\cite{AMZ+24, ZGL+25}, while the partial defense scenario is overlooked.
However, it is realistic and necessary in practice: clients with sensitive data require strong protection, while those holding low-risk or public data need few or no additional defenses~\cite{JZY+15, EHS+15, LJL+24}.
Moreover, enforcing a uniform defense strategy across all clients may unnecessarily compromise the model performance~\cite{HLY+23, BLH+24}. This arises from the inherent trade-off between privacy protection and model utility in most defense mechanisms~\cite{AMZ+24, ZHC+17, ZCK+24, ZGL+25}. 
Unlike prior studies that work on a uniform defense setting, this work focuses on a partial defense scenario.
% \textbf{Unlike previous studies that consider the uniform setting, this work investigates partial defense } and addresses the limitations of prior research in this context.

\noindent\textbf{Independent and Collaborative Modes.}
As guided by the FL protocol, each client communicates only with the central server and does not exchange any information with other clients. Likewise, existing membership inference defenses are developed individually and independently in FL. We refer to this case as the independent mode.
Beyond the independent mode, we explore a collaborative mode in which clients cooperate to defend against MIAs.
We introduce the term \textit{defender coalition} to denote a group of clients that collaboratively implement defense strategies, assuming that all clients within the coalition behave honestly and faithfully follow the designed defense framework.
Such coalition-based cooperation may occur in real-world FL, e.g., multiple devices owned by a single user in cross-device FL~\cite{KJM+16}, or institutions under shared administration in cross-silo FL~\cite{LTS+20, SME+20}.
% Moreover, such coalitions offer distinct advantages for mitigating MIAs. 
Similar to secure multi-party computation~\cite{GO+98}, this collaborative way can offer new opportunities to enhance privacy protection and maintain model utility in FL.

% We assume that clients within the coalition:
% (1) act benignly towards one another, i.e., they do not attempt to compromise each other’s data and behave honestly during training;
% (2) follow the prescribed defense protocol, executing the agreed-upon procedures faithfully.
%
% Within this coalition, each client independently executes the core defense operations while coordinating certain aspects of the defense protocol before the FL process begins. Importantly, this collaborative approach should introduce no additional privacy risks beyond those inherent in standard FL.

\noindent\textbf{Defense Objective.} 
Our defense framework is designed to achieve the following objectives:
\begin{itemize}[leftmargin=0.2in]
    \item Strong membership privacy: Mitigate MIAs launched by both the server and clients.  
    \item High model utility: Preserve model utility with performance comparable to that of undefended systems.
    \item Low overhead: Impose minimal communication and computational overhead compared with standard FL.  
    \item Robustness: Remain effective against adaptive adversaries that attempt to bypass the defense.  
\end{itemize}

% \begin{table}[h!]
%     \setlength{\abovecaptionskip}{0pt}  % 表格标题上方间距
%     \setlength{\belowcaptionskip}{0pt}  % 表格标题下方间距
%     \centering
%     \footnotesize
%     \caption{Utility loss with different numbers of defenders on CIFAR100 (10 clients in the FL setting).}
%     \label{tab:utiloss}
%     \setlength{\tabcolsep}{2pt} 
%     \begin{tabular}{lllll}
%     \hline
%     \#Defender & DPSGD\cite{AMZ+24} & Mixup\cite{ZHC+17} & HAMP\cite{ZCK+24} & GradNoise\cite{ZGL+25} \\ %Ours\\
%     \hline
%     2   & -0.01 & -0.01 & -0.03 & +0.01 \\ % -0.01 \\
%     5   & -0.05 & -0.06 & -0.10 & -0.18 \\ % -0.03 \\
%     \hline
%     \end{tabular}
%     \vspace{-0.5cm}
% \end{table}

\section{Existing Defenses and Limitations}
\label{sec:chal}

\subsection{Existing Defenses} 
% To the best of our knowledge, the partial scenario is underexplored in previous research. Therefore, we start by reviewing current defense mechanisms against MIAs in FL settings that assume a uniform scenario.
% Given that the partial defense scenario remains underexplored, we first review existing defense mechanisms designed for the uniform scenario.
We now review existing defense mechanisms that are typically designed for the uniform scenario and operate independently.

\noindent\textbf{Perturbation-based Defenses.}
Since local updates (i.e., weights or gradients) in FL can inadvertently reveal membership information to potential attackers, perturbation-based defenses are designed to mitigate this risk by obscuring these exchanged updates. 
Common strategies include employing model sparsification to limit exposed parameters~\cite{MLS+19, LJL+23}, injecting noise to mask original signals~\cite{YXF+22, BLH+24, ZGL+25}, and leveraging differential privacy techniques to provide formal privacy guarantees~\cite{AMC+16, MNJ+22}.

\noindent\textbf{Overfitting-based Defenses.}
Since overfitting is a major contributor to membership leakage~\cite{YSG+18}, overfitting-based defenses focus on narrowing the generalization gap between training and test samples to reduce the risk of MIAs. These defenses have been extensively studied in centralized settings, particularly for protecting against black-box attacks. 
Representative methods include regularization approaches (e.g., adversarial learning~\cite{NMS+18} and Mixup~\cite{ZHC+17, CZL+21}) and transfer learning techniques (e.g., knowledge distillation~\cite{SVH+21, TXM+22}).
Previous efforts to adopt centralized defenses in federated settings have failed to mitigate the threats posed by powerful trajectory-based MIAs~\cite{GYB+22, ZGL+25}.

\subsection{Limitations}
\label{subsec:limit}
We analyze the limitations of existing defenses based on two factors: ineffectiveness against emerging trajectory-based MIAs and inadequacy in handling complex non-IID data distributions.

\noindent\textbf{Limitation 1: Existing defense mechanisms are often ineffective against emerging trajectory-based MIAs.}
In FL settings, attacks have evolved beyond exploiting single model updates~\cite{NMS+19, MLS+19} to leveraging temporal information across multiple rounds, which significantly boosts the performance of advanced trajectory-based MIAs~\cite{GYB+22, ZGL+25}.
% Unlike traditional attacks that depend on a single model snapshot (e.g., final snapshot), those trajectory-based attacks monitor the evolution of model states across multiple snapshots to facilitate more advanced membership inference~\cite{ZGL+25}.
This means that historical model snapshots throughout the training process become vulnerable to potential adversaries. 
Consequently, existing overfitting-based defenses, primarily focused on reducing overfitting at convergence and protecting a single model snapshot (e.g., the final one), are insufficient.
Furthermore, perturbation-based defenses protect the entire training process by continuously adding empirical or theoretical perturbations to model updates, making it hard to balance protection strength and model utility.

\begin{figure}[t]
	\centering
	\subfigure[No defense.]{
		\centering
		\includegraphics[width=0.22\textwidth]{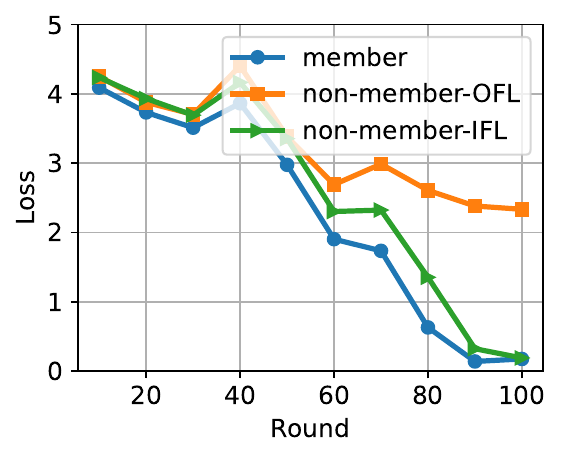}
		\label{fig:noneloss}
	}
	\subfigure[With HAMP defense~\cite{ZCK+24}.]{
		\centering
		\includegraphics[width=0.22\textwidth]{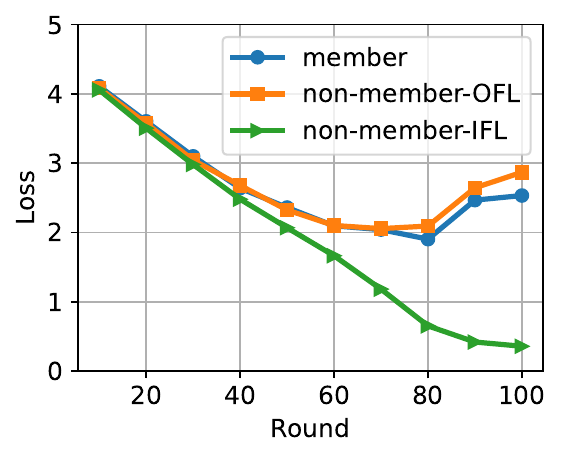}
		\label{fig:hamploss}
	}
	\caption{Loss trajectories of members and non-members.}
	\label{fig:nonmember}
\end{figure}

\noindent\textbf{Limitation 2: Existing defenses are weakened in the partial scenario due to the complexity of non-member sample distributions.} 
Existing defenses typically consider a single type of non-member samples (i.e., entirely unseen by the target model)~\cite{CDY+20, ZGL+25} and overlook the distinctive characteristics of non-member distributions in FL, which can be divided into two distinct types:
(1) \textbf{non-member-OFL}, denoting samples entirely \textit{outside} the FL system (i.e., not included in any client’s dataset); 
and (2) \textbf{non-member-IFL}, referring to samples belonging to other clients \textit{within} the FL process but not to the target client. 
This heterogeneity among non-member samples introduces significant challenges for defense design, as their behavior (e.g., loss and confidence) can vary substantially after defense implementation.
We present loss trajectories of member and non-member samples without any defense in Figure~\ref{fig:noneloss}. Non-member-IFL samples are less vulnerable, as they are trained by other clients and exhibit similar model behavior to the target client, particularly under IID settings.
However, this advantage diminishes in a partial defense scenario, where some clients actively defend while others do not. As illustrated in Figure~\ref{fig:hamploss}, while the defense effectively reduces distinguishability between member and non-member-OFL samples, it unexpectedly amplifies the difference between member and non-member-IFL samples.
Consequently, the divergence in training strategies between defended and undefended clients amplifies the distinction between member and non-member-IFL samples, which renders sensitive member data more vulnerable in partial defense scenarios.
\section{Methodology}
\label{sec:met}
The ultimate goal of membership inference defenses is to ensure that member and non-member samples are statistically indistinguishable with respect to the target model’s behavior~\cite{HLY+23}. However, the limitations discussed above underscore the inadequacy of existing approaches in effectively closing this gap. 
To address this issue, we introduce \textbf{CoFedMID} (\underline{Co}llaborative \underline{Fed}erated \underline{M}embership \underline{I}nference \underline{D}efense), a novel framework designed to defend against MIAs in the FL setting, with a specific focus on the more challenging partial defense scenario.
It is worth noting that our framework can be seamlessly applied to the uniform defense setting where all clients participate, as discussed in Section \ref{subsec:def}.
Unless otherwise specified, the term ``client'' in the following sections refers exclusively to participating members of the defender coalition.
\begin{figure}[t]
    \centering
    \includegraphics[width=0.48\textwidth]{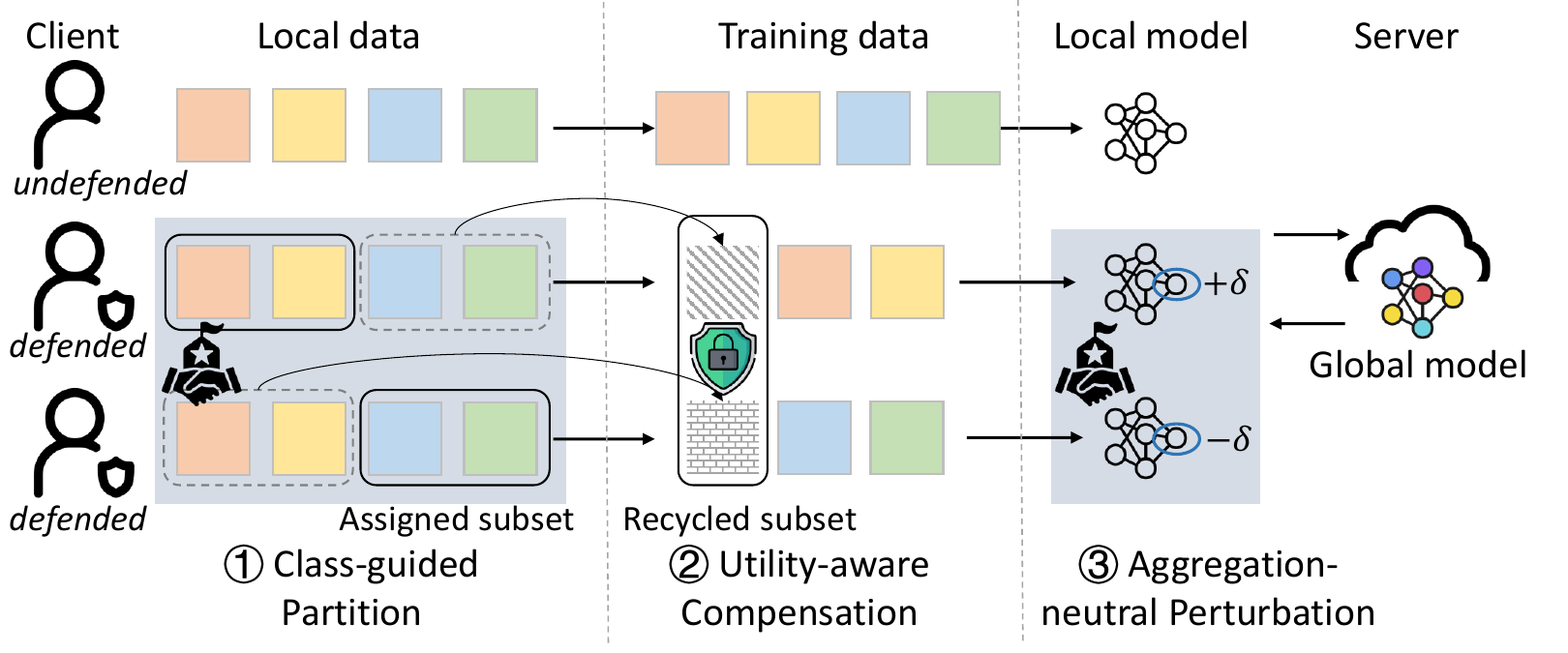}
    \caption{Overview of the proposed defense framework.
    It consists of three modules: \ding{172} class-guided partition, \ding{173} utility-aware compensation, and \ding{174} aggregation-neutral perturbation.
    Both \ding{172} and \ding{174} are performed in a collaborative mode.}
    \label{fig:overview}
\end{figure}

\noindent\textbf{Overview of CoFedMID}. The core objective of CoFedMID is to offer effective defense capabilities while preserving the overall utility of the FL system. 
This is achieved by enabling each client to update their model parameters using a carefully selected subset of local training samples, followed by the injection of strategically designed random perturbations. 
Consequently, the extent of training data memorization is substantially reduced in local models and, through federated aggregation, in the global model at each training round. 
This approach effectively mitigates vulnerability to trajectory-based MIAs (addressing Limitation 1) and decreases the distinguishability between member samples and both types of non-member samples (addressing Limitation 2). 

An overview of the framework is depicted in Figure~\ref{fig:overview}. In summary, CoFedMID comprises three key modules: \textit{class-guided partition}, \textit{utility-aware compensation}, and \textit{aggregation-neutral perturbation}. The first and third modules are executed collaboratively among clients, while the second one is performed independently by each client. We provide detailed descriptions of each module below.

\begin{remark}
Note that the first and third modules rely on a coordinator to facilitate the collaboration process. Crucially, this coordinator neither accesses any client's private data nor model parameters, nor does it impose substantial computational overhead. In our implementation, this role is randomly assigned to an arbitrary participating client. 
\end{remark}

\subsection{Class-guided Partition}
\noindent\textbf{Design Intuition}. 
The key insight underlying our method is that frequent exposure of training samples during the learning process can increase the likelihood of over-memorization~\cite{ZCB+17, SCR+17, YSG+18, CNL+19}, leading to distinct model behaviors on training samples and elevating the risk of membership privacy leakage~\cite{SCS+20, RJW+22, LRY+24}. As shown in Figure~\ref{fig:nonmember}, the gap in average loss between member samples and both types of non-member samples widens significantly after 60 rounds, regardless of whether defenses are applied.
Inspired by this observation, our defense aims to \textbf{decrease the frequency of each training sample’s usage to limit the target model’s memorization.} 
Although similar strategies like early stopping~\cite{HLY+23, ZTZ+23} have been investigated, they apply uniformly across all training samples and fail to account for sample-specific characteristics, resulting in suboptimal defense effectiveness.

We mitigate this issue by enabling clients to train their local models on a carefully selected subset of samples from specific classes during each training round. We propose a class-guided partition strategy that assigns minimally overlapping class subsets to clients, effectively limiting memorization while preserving model performance. 
This module operates through two subroutines: (1) \textit{bounded class assignment}, which allocates class subsets to clients based on refined assignment principles, and (2) \textit{decay-based assignment}, which gradually reduces the number of training samples each client uses as training progresses. 
Due to space constraints, we defer the pseudocode to Appendix \ref{sec:app-alg} (Algorithm~\ref{alg:partition}).

\noindent\textbf{Bounded Class Assignment.} 
We begin by partitioning the dataset classes among the coalition clients, where each client is assigned a unique subset of classes and restricts training to its corresponding local samples. 
To balance privacy preservation and model utility, this assignment should ensure that:
1) \textit{minimum overlap}: class subsets assigned to different clients should exhibit minimal redundancy, thereby reducing redundant exposure of training samples and mitigating membership leakage risks;
2) \textit{complete coverage}: the coalition as a whole should span the entire label space, ensuring that the learned model captures a comprehensive representation of the underlying data distribution.
Formally, let $\mathcal{C}\subseteq [K]$ denote a defender coalition consisting of $d$ clients, 
where each client $k \in \mathcal{C}$ is assigned a class subset $\mathcal{S}_k \subseteq \mathcal{S}$ with size $m$.
The goal is to determine the minimally overlapping class subsets $\{\mathcal{S}_k\}_{k=1}^{d}$ w.r.t. coalition $\mathcal{C}$ such that:
\begin{align*}
\min_{\{\mathcal{S}_k\}_{k=1}^{d}} \sum_{1 \leq i < j \leq d} |\mathcal{S}_i \cap \mathcal{S}_j|
\quad \text{subject to: }
\bigcup_{k=1}^{d} \mathcal{S}_k = \mathcal{S}, |\mathcal{S}_k| = m.
\end{align*}

To address this assignment problem, we draw inspiration from the Johnson-type bound~\cite{FPR+87, BCB+14}, 
which characterizes the theoretical minimum overlap between fixed-size subsets.  
In our setting, this leads to a lower bound on the maximum pairwise class overlap among all clients.  
Formally, the minimum achievable value of the maximum class overlap between any two clients $\lambda_{theo}$ is bounded as:
\begin{align*}
\lambda_{theo}=\max_{1 \leq i < j \leq d} |\mathcal{S}_i \cap \mathcal{S}_j| \geq \left\lceil \frac{d m^2 - N m}{N(d-1)} \right\rceil.
\end{align*}

Motivated by this theoretical bound, we develop a heuristic algorithm that approximates near-optimal class-to-client assignments, allocating classes iteratively while controlling pairwise overlaps across the coalition.
The assigned training subset $D_k^{\text{ass}}$ for client $k$ is then constructed by collecting all local samples whose labels belong to the allocated class subset $\mathcal{S}_k$ as follows:
\begin{align*}
D_k^{\text{ass}} = \{ (\x,\y) \mid (\x,\y) \in D_k, \ \y_c \in \mathcal{S}_k \}.
\end{align*}

\noindent\textbf{Decay-based Assignment.} 
Next, we dynamically adjust the class subset size $m$ throughout the training process to further limit the exposure of training samples. Prior studies~\cite{BYL+09, ZYL+18} have demonstrated that, during the early stages of training, clients benefit from diverse data to effectively capture the underlying distribution. As the global model gradually converges and stabilizes, fewer samples are sufficient to preserve its performance. Building on this insight, we employ a linear decay function that reduces $m$ from \( m_{\max} \) to \( m_{\min} \) at a rate of $\alpha$ over the total number of training rounds $T$, that is, 
{\small
\begin{align*}
m^t = \max \left( m_{\min}, \; m_{\max} - \gamma (t-1) \right) \text{ where } \gamma=\frac{m_{\max} - m_{\min}}{T - 1}.
\end{align*}
}%
Here $m^t$ denotes the class subset size at round $t\in [T]$.
Nonlinear strategies, including exponential and polynomial decay, are also investigated in our experiments (see Section \ref{subsec:details}).

The bounded class assignment is implemented by the coordinator, who knows the label-space IDs and assigns classes to members such that their union collectively covers the entire label space.
Since our assignment centers on the coalition's overall class coverage, the coordinator does not require access to the specific class distribution of any individual member.
Each coalition member uses assigned classes if they have corresponding samples; otherwise, the assignments are ignored. 
% analysis
We highlight that the effectiveness of our class-guided partition module arises not only from its core objective of minimizing the exposure of each client’s local training data, but also from its deliberate disruption of the implicit mapping between local datasets and their resulting model parameters. This is achieved through our randomized, adaptive class assignment procedure that increases uncertainty for potential adversaries attempting to infer whether a target sample belongs to a client, hence further amplifying the membership privacy protection. 

\subsection{Utility-aware Compensation}

%Problem
\noindent\textbf{Design Intuition.} 
While limiting the exposure of training samples to local models can enhance the defense against MIAs, it often comes at the cost of degraded model utility due to the limited data utilization.
This trade-off becomes more pronounced as the size of the defender coalition grows. 
As demonstrated in Appendix~\ref{app:usage}, a larger number of participating clients results in a smaller portion of training data being utilized, thereby leading to a significant utility loss.

% Idea
To restore model performance, we propose selectively using a subset of samples that were previously excluded by the participating client in the coalition, referred to as \textit{recycled samples}, while carefully managing their risks of membership privacy leakage.
To this end, we introduce a utility-aware compensation module that consists of two components: (1) \textit{sample recycling}, which identifies the recycled samples at specified training rounds that significantly improve the utility of the corresponding local models,
and (2) \textit{confidence regularization}, which constrains overconfidence on these recycled samples to control privacy risks. The complete procedure is provided in the Appendix Algorithm~\ref{alg:compensation}.

% Defenses
\subsubsection{Sample Recycling}
We now present a high-level overview of our sample recycling process. 
To effectively mitigate utility degradation, two pivotal considerations must be addressed: when to recycle and which samples to recycle. 
Recall that as the global model converges, the number of classes $m$ assigned to each client gradually decreases, leading to a smaller portion of training samples being selected. Based on this observation, we recommend starting sample recycling at an intermediate stage $t'$ rather than from the outset. 
A practical guideline for determining $t'$ is to monitor $m$ and trigger recycling once it falls below a predefined threshold.

For the first consideration, we propose selecting recycled samples according to their ``compensation'' for utility degradation, 
which can be quantified by the magnitude of \textit{loss reduction} evaluated on the updated local model upon reusing these samples.
This is inspired by curriculum learning techniques~\cite{BYL+09, WXC+21}, in which the model is trained progressively by first learning from easier samples with lower losses to achieve stability, followed by learning harder ones with larger losses to enhance generalization. 
Specifically, we introduce the notion of a \textit{sample interval} which refers to a subset of samples exhibiting comparable levels of learning difficulty. 
At each training round $t\in[t', T]$, every client $k$ in the coalition partitions its local dataset $D_k$ into $M$ disjoint sample intervals, denoted as $\mathcal{I}_k = \{ \mathcal{I}^1_k, \mathcal{I}^2_k, \ldots, \mathcal{I}^M_k \}$\footnote{For brevity, we simplify the notations by omitting the round identifier $t$, and will continue to do so throughout the remainder of this section.}.
Guided by a carefully designed recycling strategy, each client then selects an optimal interval $\mathcal{I}^{j}_k$ that maximizes utility compensation, and the set of recycled samples $D_k^{\text{rec}}$ is formed as the set difference $\mathcal{I}^{j}_k \setminus D_k^{\text{ass}}$.

The key technical challenges involved in implementing this approach are twofold: 1) constructing sample intervals that precisely capture the learning difficulty of training samples; and 2) establishing a robust strategy for selecting the optimal interval. We tackle these challenges through the following three steps, with Step 1 addressing the first challenge and Steps 2-3 addressing the second.

\textbf{Step 1: Sample Interval Initialization.} 
The construction of sample intervals starts by evaluating the per-sample loss values using the initial local parameters (i.e., the global model obtained from the preceding round).
Considering the distribution of these losses varies across training rounds, we apply min-max normalization to ensure that the partitioning reflects the relative difficulty of samples in the current training context.
Formally, given the model parameters $\tilde{\theta}$, we denote 
$\ell_i$ as the original loss value of the $i$-th training sample in $D_k$. 
Then the normalized loss value is calculated as
$
\tilde{\ell}_i = \frac{\ell_i - \ell_{\min}}{\ell_{\max} - \ell_{\min}}, 
$
where \(\ell_{\min}\) and \(\ell_{\max}\) are the minimum and maximum values among the original losses, respectively. Next, we sort $D_k$ by normalized losses in ascending order and divide it into $M$ intervals of approximately equal size. 
Consequently, the sample intervals can be initialized as:
\[
\mathcal{I}_k^j = \left[ \tilde{\ell}_{(q_{j-1})},\; \tilde{\ell}_{(q_j)} \right), \quad \text{for } j = 1, 2, \ldots, M,
\]
where \(\tilde{\ell}_{(q_j)}\) denotes the \(q_j\)-th smallest normalized loss value,
with \(q_j = \lfloor (j/M) \cdot |D_k| \rfloor\), \(q_0 = 0\), and \(q_M = |D_k|\).

This initialization guarantees that each interval initially contains an equal number of samples, thereby eliminating quantity-induced bias and facilitating fair exploration in the early learning stages. Furthermore, the design allows the number of samples per interval to be dynamically adjusted over time, reducing the recycled sample size in later stages to limit additional exposure.

\textbf{Step 2: Sample Interval Evaluation.}
Once the sample intervals are constructed, the next step is to develop a sample recycling strategy that selects the optimal sample interval. 
As previously noted, our intuition is that the contribution of a training sample to utility improvement can be quantified by the reduction in loss values observed between local models trained with and without that sample~\cite{GAB+17, TYX+24}. We treat this loss reduction as a reward signal to guide the selection of sample intervals throughout the training process.

%$D_k^{\text{rec}} = \mathcal{I}^{j}_k \setminus D_k^{\text{ass}}$
Specifically, given an arbitrary sample interval $\mathcal{I}^{j}_k$ and its corresponding set of recycled samples $D_k^{\text{rec}}$, we denote $\theta_k^{\prime}$ as the initial local model and denote $\theta_k$ as the updated local model trained on \(D_k^{\text{rec}} \cup D_k^{\text{ass}}\).
Then the reward value $r$ w.r.t. the interval is defined as the average loss reduction computed over a validation set \(D_k^{\text{val}}\):
\begin{align}
\label{eq:ori-reward}
r = \frac{1}{|D_k^{\text{val}}|} \sum_{(\x, \y) \in D_k^{\text{val}}} \left[ \ell(f_{\theta_k^{\prime}}(\x),\y) - \ell(f_{\theta_k}(\x),\y) \right].
\end{align}
We attribute this reduction to \(D_k^{\text{rec}}\) as our recycling strategy prioritizes high-loss samples, which are expected to contribute more to model performance than \(D_k^{\text{ass}}\) composed of samples with varying loss levels~\cite{KAF+18}.
In line with previous studies~\cite{GAB+17, TYX+24}, we rescale the range of raw rewards to \([-1, 1]\) using the 20th and 80th percentiles of historical rewards. The normalized reward \(\tilde{r}\) is defined as 
$\tilde{r} = \max(-1, \min(1, \frac{2(r - r_{20})}{r_{80} - r_{20}} - 1))$.
This transformation smooths out extreme fluctuations and stabilizes the learning signal during interval selection.

\textbf{Step 3: Sample Interval Selection.}
In this step, our goal is to select the optimal sample interval by maximizing the reward values defined above. This process inherently involves a core trade-off between \textit{exploitation} (i.e., choosing intervals that have previously led to substantial loss reductions) and \textit{exploration} (i.e., exploring alternative intervals that may yield greater reductions).
We model this selection process as a multi-armed bandit (MAB) problem~\cite{MAT+08}, where each ``arm'' represents a sample interval.
We adopt the EXP3~\cite{APC+02, TYX+24} algorithm, which balances exploration and exploitation by assigning a carefully designed weight to each arm based on its observed reward, enabling adaptation to non-stationary and noisy reward distributions during training.
Upon receiving a normalized reward, the algorithm updates the weight of the selected arm based on the reward and its sampling probability.
Moreover, we control recycling volume when an interval contains many less informative samples, especially in the later stages. 
When the candidate samples exceed a predefined threshold $r_l$, random subsampling is applied to prevent overexposure to easy examples and ensure training resources focus on more valuable data.

\subsubsection{Confidence Regularization}
Intuitively, sample recycling tends to select data points with high loss gaps—often those near decision boundaries or behaving as outliers~\cite{LZZ+21, CCC+21}, which are consequently more vulnerable to MIAs. Thus, to mitigate this increased vulnerability, we regularize the model’s confidence on recycled samples by encouraging smoother output distributions and reducing overconfident predictions.
Following~\cite{ZCK+24}, for each recycled sample $(\x,\y) \in D_k^{\text{rec}}$, we construct a soft label $\tilde{\y}\in\mathbb{R}^N$ by assigning the predicted confidence $p_c=(f_{\theta_k}(\x))_c$ for the ground-truth class $\y_c$ while evenly distributing the residual probability mass $1 - p_c$ across the other $N-1$ classes.
Together with an entropy-based constraint, the confidence-regularized loss function is formulated as follows:
\begin{align*}
\ell_{cr}(f_{\theta_k}(\x),\y) \triangleq \mathrm{KL}(\log f_{\theta_k}(\x) \, \| \, \tilde{\y}) - \mu \cdot \mathrm{Entropy}(f_{\theta_k}(\x)),
\end{align*}
where \(\mu\) is a hyperparameter balancing the two terms.  
Together with the standard loss function $\ell(\cdot)$, the training objective for client \(k\) is to minimize:
{\small
\begin{align*}
% \label{eq:loss}
\mathcal{L}_k(\theta_k;D_k) \triangleq \sum_{(\x,\y) \in D_k^{\text{rec}} \cup D_k^{\text{ass}}} \ell(f_{\theta_k}(\x), \y) + \sum_{(\x,\y) \in D_k^{\text{rec}}} \ell_{cr}(f_{\theta_k}(\x),\y).
\end{align*}
}

In summary, we propose a utility-aware compensation algorithm designed to restore model performance by selectively reusing a minimal subset of previously excluded samples. We employ a multi-armed bandit framework to dynamically prioritize sample intervals that offer the greatest utility improvement and incorporate a specialized regularizer to mitigate privacy risks associated with the recycled samples. 
The above process uniformly treats all training samples, classes, and protected clients.
Due to dynamic class assignment among different rounds, where no fixed client-class mapping exists, samples can originate from either assigned classes or serve as recycled data.

\subsection{Aggregation-neutral Perturbation}

\textbf{Design Intuition.} 
% Compared to attacks on the aggregated global model, MIAs targeting local models pose a more severe threat, as adversaries can directly access the local model parameters trained on private local datasets.
Unlike the aggregated global model, local models are more vulnerable as they directly expose client-specific information to potential attackers.
Inspired by multi-party secure computation~\cite{GO+98, MVP+19}, we introduce an aggregation-neutral perturbation module to further strengthen the overall defense capability by perturbing local models before aggregation. Specifically, this module injects carefully crafted noise into the local parameters of clients within the defender coalition $\mathcal{C}$, effectively masking private information while preserving the correctness of global aggregation.
To achieve this, the injected noise is constructed to meet two key criteria: (1) its weighted sum across all collaborative clients is zero, and (2) it is applied locally to only a small subset of model parameters. The complete procedure is presented in the Appendix Algorithm~\ref{alg:perturb}.

\noindent\textbf{Noise Assignment.}
To satisfy the first criterion, we design the client noises to be orthogonal to the aggregation weight vector in \( \mathbb{R}^d \). We first generate a preliminary Gaussian noise \( \delta'_k \sim \mathcal{N}(0, \sigma^2 I) \) for each client $k\in\mathcal{C}$ independently, where $\sigma$ is the perturbation strength.
Given the aggregation weights of all clients \( \mathbf{w} = [w_1, \dots, w_d] \), the scalar noise \( \delta_k \) for client \( k \) is generated by projecting a base noise \( \delta'_k \) onto the direction orthogonal to \( \mathbf{w} \):
\begin{align*}
\delta_k = \delta'_k - \frac{w_k}{\|\mathbf{w}\|_2^2} \sum_{j=1}^{d} w_j \delta'_j, \quad \text{for } k = 1, \ldots, d.
\end{align*}
This construction guarantees that the perturbations cancel out during aggregation, that is, \( \sum_{k=1}^d w_k \delta_k = 0 \).

\noindent\textbf{Model Perturbation.}
To satisfy the second criterion, perturbations are selectively applied to a subset of model parameters (or gradients) so as to guarantee their defensive efficacy while incurring small disruption.
Our approach perturbs only the final few layers while keeping earlier layers unchanged. This is because deeper layers (e.g., the classification layer) closer to the output exert a greater influence on predictions~\cite{HSP+15, ZBK+16}.
Specifically, we begin by flattening all model parameters into a one-dimensional vector, arranged sequentially by layer. For a perturbation ratio \(r_P \in (0, 1.0]\), each client perturbs the last \( P^\prime \triangleq \lfloor r_P \cdot P \rfloor \) parameters as
\begin{align*}
\theta_k^{(i)} \gets \theta_k^{(i)} + \delta_k, \quad \text{where } i\in[P - P^\prime + 1,P].
\end{align*}

Since clients may not know their exact aggregation weights \(w_k\), a practical approach is to approximate them by the size of each client’s local dataset, assuming that aggregation weights are proportional to data volume. For FedAvg, we compute weights based on each client’s local training data volume rather than assigned classes. This reduces communication overhead, as members do not need to report data size every round, and helps hide the defense coalition from an honest-but-curious server.

\subsection{Algorithm Summary}
We present the complete pseudocode of CoFedMID in the Appendix Algorithm~\ref{alg:overall}.
Prior to training, a randomly chosen client serves as the coordinator, allocating data classes and configuring noise perturbations for all defended members.
Each client is then assigned a series of class subsets with controlled overlap, which is dynamically reduced each round following a decay schedule.
Subsequently, during each training round, clients train on partial local data consisting of assigned and recycled samples, balancing utility preservation with privacy protection. Moreover, their model parameters are then perturbed with allocated neutral noise to obscure individual training behaviors.
Meanwhile, clients outside the coalition conduct standard local training. The server ultimately aggregates all updated parameters via weighted averaging.
Additionally, while our design focuses on a single defender coalition, it can be extended to multiple ones, with each employing the proposed defense framework independently.

Our framework operates within a coalition formed by clients collaborating on defense, and specifically, such a coalition can arise in two practical scenarios:
(a) Coalition with prior relationships (our focus): Clients managed by the same supermaster (e.g., a user’s phone and iPad participating in Gboard training via FL) or those that have interacted previously (e.g., devices joining the same IoT network) can autonomously form a trustworthy coalition~\cite{ASW+23}. The coordinator is randomly selected from these clients through a decentralized random election mechanism~\cite{ZNM+24}.
(b) Coalition without prior relationships (general case): Clients can request to join a coalition through a trusted third party~\cite{CSL+22}, who acts as the coordinator and manages the coalition formation. Before FL begins, each coalition member reports its local data size to the coordinator.
For coordination, scenario (a) can leverage stable pseudonymous identities (e.g., a client index~\cite{KDK+23}) while scenario (b) requires no client coordination at all. Consequently, in both scenarios, all subsequent communications between any client and the coordinator, including class assignments, data volumes, and scalar noise, do not access private data or model parameters.
\section{Evaluation}
\label{sec:exp}

\subsection{Experiment Setup}
\textbf{Datasets and Models.}
Following previous works on MIA research~\cite{WRL+24, ZCK+24}, our experiments focus on image classification and are conducted on three benchmark datasets: CIFAR10~\cite{KAH+09}, CIFAR100~\cite{KAH+09}, and TinyImageNet~\cite{YLX+15} on ResNet18~\cite{HKZ+16} and modified WideResNet-16-4~\cite{AMZ+24}. 
Their performance in various datasets and model architectures is demonstrated in Appendix~\ref{app:results}.

\noindent\textbf{Federated Setting.}
Unless otherwise specified, we consider a horizontal FL system comprising $K=10$ clients under an IID data distribution. 
We employ FedAvg~\cite{MMR+17}, where each client performs one local training epoch per round for up to 100 training rounds. The aggregation weight for each client is set proportional to the size of their local dataset. We direct readers to Section \ref{subsec:def} for extended scenarios involving a varying number of clients and non-IID data distributions.

\begin{table*}[h!]
    \setlength{\abovecaptionskip}{0pt}  % 表格标题上方间距
    \setlength{\belowcaptionskip}{0pt}  % 表格标题下方间距
    \centering
    \footnotesize
    \caption{AUC results of defense methods against different attacks on CIFAR100 with ResNet18.
    Lower values mean better defense.
    % \colorbox{gray!20}{Grey blocks} mark the best average results, and \underline{underlining} denotes the second-best.
    \textbf{Avg} is the average results over seven MIAs, and $\Delta$\textbf{Acc} denotes the test accuracy change relative to the undefended FL (original accuracy).
    ‘Pair’ denotes a defender coalition of two clients, whereas ‘Half’ refers to a coalition comprising half of the clients in the FL system.
    The best average results are highlighted in \textbf{bold}.
    % xw{Can combine Tabel2 and Table3. The $\pm $ 0.02 is not important here, can be deleted.}
    }
    \label{tab:main_res18_c100}
    \begin{tabular}{l|l|ccccccc|cc}
    \toprule
    \textbf{Case} & \textbf{Defense} & \textbf{Loss-Series} & \textbf{Avg-Cosine} & \textbf{FedMIA-I} & \textbf{FedMIA-II} & \textbf{FTA-C} & \textbf{FTA-L} & \textbf{SeqMIA} & \textbf{Avg} $\downarrow$ & $\Delta$\textbf{Acc} $\uparrow$ \\
    \midrule
    \multirow{1}{*}{} 
       - & No Defense & $0.65 \pm 0.02$ & $0.78 \pm 0.02$ & $0.72 \pm 0.04$ & $0.82 \pm 0.00$ & $0.71 \pm 0.01$ & $0.80 \pm 0.01$ & $0.89 \pm 0.02$ & 0.77 & (0.46) \\
        \midrule
        
    \multirow{7}{*}{Pair} 
        & GradSparse & $0.69 \pm 0.03$ & $0.81 \pm 0.02$ & $0.80 \pm 0.05$ & $0.86 \pm 0.03$ & $0.76 \pm 0.04$ & $0.82 \pm 0.02$ & $0.94 \pm 0.03$ & 0.78 & -0.01 \\
        & GradNoise  & $0.71 \pm 0.00$ & $0.81 \pm 0.00$ & $0.72 \pm 0.01$ & $0.83 \pm 0.00$ & $0.77 \pm 0.01$ & $0.83 \pm 0.01$ & $0.90 \pm 0.01$ & 0.79 & +0.01 \\
        & DPSGD & $0.54 \pm 0.00$ & $0.54 \pm 0.00$ & $0.50 \pm 0.01$ & $0.53 \pm 0.00$ & $0.50 \pm 0.00$ & $0.54 \pm 0.01$ & $0.81 \pm 0.02$ & 0.57 & -0.01 \\
        & Mixup & $0.53 \pm 0.00$ & $0.73 \pm 0.00$ & $0.51 \pm 0.00$ & $0.65 \pm 0.00$ & $0.52 \pm 0.01$ & $0.60 \pm 0.00$ & $0.50 \pm 0.03$ & 0.57 & -0.01 \\
        & RelaxLoss & $0.65 \pm 0.00$ & $0.68 \pm 0.00$ & $0.73 \pm 0.00$ & $0.83 \pm 0.00$ & $0.68 \pm 0.01$ & $0.78 \pm 0.01$ & $0.65 \pm 0.03$ & 0.71 & -0.01 \\
        & HAMP & $0.61 \pm 0.00$ & $0.65 \pm 0.00$ & $0.73 \pm 0.00$ & $0.62 \pm 0.00$ & $0.51 \pm 0.01$ & $0.54 \pm 0.01$ & $0.55 \pm 0.02$ & 0.55 & -0.03 \\
        & \textbf{CoFedMID} & $0.54 \pm 0.00$ & $0.50 \pm 0.01$ & $0.52 \pm 0.01$ & $0.57 \pm 0.01$ & $0.51 \pm 0.01$ & $0.58 \pm 0.01$ & $0.54 \pm 0.05$ &  \textbf{0.52} & -0.01 \\ %\cellcolor{gray!20}
        \midrule

    \multirow{7}{*}{Half} 
        & GradSparse & $0.69 \pm 0.03$ & $0.80 \pm 0.02$ & $0.80 \pm 0.04$ & $0.86 \pm 0.03$ & $0.75 \pm 0.04$ & $0.82 \pm 0.02$ & $0.94 \pm 0.02$ & 0.78 & -0.01 \\
        & GradNoise & $0.62 \pm 0.11$ & $0.70 \pm 0.13$ & $0.62 \pm 0.17$ & $0.68 \pm 0.16$ & $0.62 \pm 0.14$ & $0.67 \pm 0.17$ & $0.74 \pm 0.20$ & 0.66 & -0.18 \\
        & DPSGD & $0.56 \pm 0.00$ & $0.55 \pm 0.01$ & $0.57 \pm 0.00$ & $0.53 \pm 0.00$ & $0.53 \pm 0.01$ & $0.53 \pm 0.00$ & $0.92 \pm 0.01$ & 0.59 & -0.05 \\
        & Mixup & $0.52 \pm 0.00$ & $0.66 \pm 0.00$ & $0.50 \pm 0.00$ & $0.64 \pm 0.00$ & $0.54 \pm 0.01$ & $0.57 \pm 0.00$ & $0.80 \pm 0.03$ & 0.57 & -0.06 \\
        & RelaxLoss & $0.66 \pm 0.00$ & $0.66 \pm 0.00$ & $0.72 \pm 0.00$ & $0.83 \pm 0.00$ & $0.66 \pm 0.01$ & $0.77 \pm 0.01$ & $0.69 \pm 0.02$ & 0.71 & -0.00 \\
        & HAMP & $0.56 \pm 0.00$ & $0.62 \pm 0.00$ & $0.70 \pm 0.00$ & $0.61 \pm 0.00$ & $0.55 \pm 0.01$ & $0.53 \pm 0.00$ & $0.77 \pm 0.08$ & 0.59 & -0.10 \\
        & \textbf{CoFedMID} & $0.50 \pm 0.00$ & $0.50 \pm 0.01$ & $0.54 \pm 0.01$ & $0.51 \pm 0.00$ & $0.52 \pm 0.01$ & $0.56 \pm 0.02$ & $0.59 \pm 0.03$ &  \textbf{0.51} & -0.03 \\

    \bottomrule
    \end{tabular}
\end{table*}

\begin{table*}[h!]
    \setlength{\abovecaptionskip}{0pt}  % 表格标题上方间距
    \setlength{\belowcaptionskip}{0pt}  % 表格标题下方间距
    \centering
    \footnotesize
    \caption{TF01 results of defense methods against different attacks on CIFAR100 with ResNet18.
    }
    \label{tab:main_res18_c100_TF01}
    \setlength{\tabcolsep}{5pt} % 默认是6pt
    \begin{tabular}{l|l|ccccccc|cc}
    \toprule
    \textbf{Case} & \textbf{Defense} & \textbf{Loss-Series} & \textbf{Avg-Cosine} & \textbf{FedMIA-I} & \textbf{FedMIA-II} & \textbf{FTA-C} & \textbf{FTA-L} & \textbf{SeqMIA} & \textbf{Avg} $\downarrow$ & $\Delta$\textbf{Acc} $\uparrow$ \\
    \midrule
    \multirow{1}{*}{} 
        - & No Defense & $10.52 \pm 1.02$ & $6.44 \pm 2.09$ & $4.89 \pm 2.55$ & $12.36 \pm 2.05$ & $2.22 \pm 1.13$ & $2.20 \pm 0.67$ & $10.62 \pm 4.52$ & 6.59 & (0.46) \\
        \midrule
        
    \multirow{7}{*}{Pair} 
        & GradSparse & $12.55 \pm 1.40$ & $10.87 \pm 4.11$ & $10.03 \pm 3.97$ & $20.13 \pm 10.76$ & $1.49 \pm 1.07$ & $3.52 \pm 0.66$ & $16.48 \pm 6.48$ & 7.87 & -0.01 \\
        & GradNoise  & $13.82 \pm 0.00$ & $8.82 \pm 0.82$ & $3.51 \pm 0.51$ & $13.72 \pm 0.40$ & $2.30 \pm 0.21$ & $3.94 \pm 0.34$ & $14.69 \pm 5.00$ & 8.68 & +0.01 \\
        & DPSGD & $15.62 \pm 0.45$ & $3.08 \pm 0.14$ & $5.77 \pm 0.22$ &  $0.31 \pm 0.02$ & {$1.93 \pm 0.27$} & $0.06 \pm 0.06$ & $4.58 \pm 2.36$ & 4.05 & -0.01 \\
        & Mixup & $12.00 \pm 0.00$ & $4.44 \pm 0.42$ & $13.64 \pm 1.10$ & $16.72 \pm 0.75$ & $2.84 \pm 0.22$ & $0.89 \pm 0.12$ & {$0.15 \pm 0.08$} & 5.81 & -0.01 \\
        & RelaxLoss & $14.47 \pm 0.34$ &  $1.67 \pm 0.06$ & $3.50 \pm 0.12$ & $10.82 \pm 0.21$ & $3.96 \pm 0.17$ & $2.26 \pm 0.38$ & $8.49 \pm 1.99$ & 5.33 & -0.01 \\
        & HAMP & {$9.42 \pm 0.00$} & $18.73 \pm 0.45$ &  $0.03 \pm 0.02$ & $21.12 \pm 0.37$ &  $0.00 \pm 0.00$ & $0.75 \pm 0.09$ & $7.81 \pm 3.34$ & 7.74 & -0.03 \\
        & \textbf{CoFedMID} &  $2.28 \pm 0.43$ & {$2.08 \pm 0.44$} & {$0.97 \pm 0.28$} & {$3.12 \pm 0.80$} & $2.18 \pm 1.25$ & {$0.29 \pm 0.19$} &  $0.11 \pm 0.10$ &  \textbf{1.58} & -0.01 \\
        \midrule

    \multirow{7}{*}{Half} 
        & GradSparse & $13.23 \pm 1.04$ & $9.26 \pm 3.30$ & $13.08 \pm 5.89$ & $19.18 \pm 8.50$ & $8.29 \pm 2.63$ & $2.23 \pm 0.14$ & $13.80 \pm 5.18$ & 11.87 & -0.01 \\
        & GradNoise & {$10.65 \pm 7.92$} & $5.21 \pm 4.82$ & $6.36 \pm 2.30$ & $12.50 \pm 5.32$ & $3.83 \pm 3.08$ & $1.14 \pm 1.08$ & $10.48 \pm 11.00$ & 7.45 & -0.18 \\
        & DPSGD & $16.93 \pm 0.32$ & $3.10 \pm 0.00$ & $1.82 \pm 0.29$ &  $0.20 \pm 0.08$ &  $2.34 \pm 0.65$ &  $0.00 \pm 0.00$ & $15.73 \pm 3.65$ & 5.59 & -0.05 \\
        & Mixup & $13.63 \pm 0.12$ & $4.38 \pm 0.22$ & $8.28 \pm 0.34$ & $14.89 \pm 0.44$ & $4.02 \pm 0.19$ & $0.10 \pm 0.06$ & {$1.20 \pm 0.60$} & 4.75 & -0.06 \\
        & RelaxLoss & $16.43 \pm 0.07$ &  $0.28 \pm 0.06$ & $3.52 \pm 0.48$ & $12.48 \pm 1.25$ & {$2.94 \pm 0.46$} & $1.82 \pm 0.21$ & $8.55 \pm 1.92$ & 3.98 & -0.00 \\
        & HAMP & $11.18 \pm 0.00$ & $13.58 \pm 0.25$ &  $0.00 \pm 0.00$ & $15.00 \pm 0.38$ & $6.33 \pm 0.44$ & $0.20 \pm 0.08$ & $13.15 \pm 5.07$ & 5.46 & -0.10 \\
        & \textbf{CoFedMID} &  $3.34 \pm 0.72$ & {$2.54 \pm 1.09$} & {$0.24 \pm 0.11$} & {$0.48 \pm 0.17$} & $4.18 \pm 1.29$ & {$0.06 \pm 0.05$} &  $0.71 \pm 0.53$ &  \textbf{1.65} & -0.03 \\
        
    \bottomrule
    \end{tabular}
\end{table*}

\noindent\textbf{Attack Algorithms.}
We implement seven trajectory-based attacks from both client-side and server-side perspectives, including Loss-Series~\cite{YSG+18}, Avg-Cosine~\cite{LJL+23}, FedMIA-I/II~\cite{ZGL+25}, FTA-C/L~\cite{CHE+24}, and SeqMIA~\cite{LHL+24}.
These attacks exploit temporal information, i.e., the changes in model updates or outputs across training rounds, to infer the membership status of query samples.
Further details are provided in the Appendix~\ref{app:attack}.

\noindent\textbf{Baseline Defenses.}
We compare our framework with three typical FL defenses, including %like 
gradient sparsification (GradSparse)~\cite{GOR+18} and noise injection (GradNoise)~\cite{BLH+24, ZGL+25}, DPSGD~\cite{AMZ+24}. Additionally, we include comparisons with three defenses originally designed for centralized settings: Mixup~\cite{ZHC+17, GHL+23}, Relaxloss~\cite{CDY+22}, and HAMP~\cite{ZCK+24}. 
The Appendix~\ref{app:defense} provides a detailed description of these defenses and their configurations.

\noindent\textbf{CoFedMID Configurations.}  
We set $m_{\text{max}}$ to 10 for CIFAR10, 50 for CIFAR100, and 200 for TinyImageNet.
Correspondingly, $m_{\text{min}}$ is set to 2, 20, and 40, respectively (i.e., 20\% of the total number of classes).
Across all three datasets, the compensation module is triggered from round 10 (i.e., $t^{\prime} = 10$), with $M$ set to 10 and $\mu$ to 0.005~\cite{ZCK+24}.
For the perturbation module, the perturbation strength and ratio are set to 0.1 and 0.2 for CIFAR100, with the ratio increased to 0.3 for TinyImageNet, and to 0.01 and 0.01 for CIFAR10, respectively.

\noindent\textbf{Evaluation Metrics.}
Following~\cite{CNC+22,WRL+24,ZCK+24}, we adopt two common MIA evaluation metrics: AUC, the area under the ROC curve (0.5 indicates random guessing), and TPR (True Positive Rate) at a low FPR (False Positive Rate), which measures the ability of an attack to correctly identify members (lower is better). In our experiments, we use TF01 to denote TPR(\%)@FPR=0.1\%.

All experiments were conducted on an Ubuntu system with NVIDIA RTX 3090 GPUs, and the average results are reported over five runs with different random seeds.

\subsection{Defense Performance}
\label{subsec:def}
This subsection mainly evaluates the defense performance of CoFedMID in two representative coalition settings: 
(1) \textit{Pair}: a small coalition of two clients, and (2) \textit{Half}: a large coalition comprising half of all clients.
We then assess its effectiveness in a variety of settings, including non-IID data distributions, different numbers of clients, and a uniform defense scenario.

\noindent \textbf{Performance in the Partial Defense Scenario.}
We begin by analyzing the minimal coalition required for collaborative defense in the pair-client setting.
We present a comparison of all defense methods on CIFAR100 with ResNet18 using AUC and TF01 metrics in Tables~\ref{tab:main_res18_c100} and \ref{tab:main_res18_c100_TF01}.
Additional results for other datasets and models are provided in Appendix~\ref{app:results}.
Overall, CoFedMID consistently achieves superior or comparable defense performance across all attacks, yielding the lowest average values on both AUC and TF01 metrics.
In contrast, baseline methods may be effective against certain attacks but perform significantly worse against others.
For example, although HAMP attains the second-best defense overall, its performance is notably inconsistent, exhibiting substantially poorer results than CoFedMID on Avg-Cosine and FedMIA-I, with AUC values approximately 0.15 and 0.25 higher, respectively. 
Furthermore, we provide a micro-level MIA analysis to evaluate our framework’s uniform defense, examining its performance at the granularity of individual classes and clients in Appendix~\ref{app:microays}.

% \noindent \textbf{Performance with a large-scale coalition.}
We then analyze defense performance under a larger defender coalition, i.e., the Half case.
A key observation is that most baseline methods experience significant accuracy degradation as the number of defended clients increases. For example, baselines with AUCs below 0.60 show accuracy drops of 0.05–0.10, substantially greater than those in the Pair case.
In contrast, CoFedMID maintains defense performance comparable to the Pair setting, while incurring only a 0.03 reduction in accuracy.
This underscores the scalability and effectiveness of our framework in different partial scenarios.

% next experiment use two representative MIAs
\begin{table}[htbp]
    \setlength{\abovecaptionskip}{0pt}  % 表格标题上方间距
    \setlength{\belowcaptionskip}{0pt}  % 表格标题下方间距
    \centering
    \footnotesize
    \caption{Comparison under the uniform defense scenario on CIFAR100 (AUC/TF01).}
    \label{tab:uniform}
    \begin{tabular}{l|llll}
        \toprule
        MIA\&Acc     & DPSGD &  Mixup & HAMP & CoFedMID \\
        \midrule
        FedMIA-II    & 0.53/1.44 & 0.58/7.41   & 0.64/3.43 & 0.54/0.51 \\
        SeqMIA       & 0.98/46.79 & 0.98/25.46  & 0.99/59.7 & 0.69/0.56 \\
        $\Delta$Acc  & -0.16 & -0.21        & -0.14     & -0.04 \\ % 0.43 从c100开始
        % Test  & 0.46         & 0.30 & 0.25        & 0.32      & 0.43 \\ % 0.43 从c100开始
        \bottomrule
    \end{tabular}
\end{table}

\noindent \textbf{Performance in the Uniform Defense Scenario.}
Although our framework focuses on the partial defense scenario, it can be seamlessly extended to the uniform defense setting where all clients participate. 
To evaluate this, we compare the performance of CoFedMID with three representative baselines selected for their defense capabilities. In addition, we focus on two advanced attacks, FedMIA-II and SeqMIA, due to their high attack strength and their reliance on distinct forms of temporal information.
As shown in Table~\ref{tab:uniform}, our framework achieves the lowest attack performance on both metrics, at the cost of only a slight reduction in utility.
% These results highlight its effectiveness both in partial and uniform scenarios.

\begin{figure}[t]
    \setlength{\subfigcapskip}{0pt} % 图题与图片的距离
    \setlength{\subfigtopskip}{0pt} % 上方空白
    \setlength{\subfigbottomskip}{0pt} % 下方空白
	\centering
	\subfigure{
		\centering
		\includegraphics[width=0.20\textwidth]{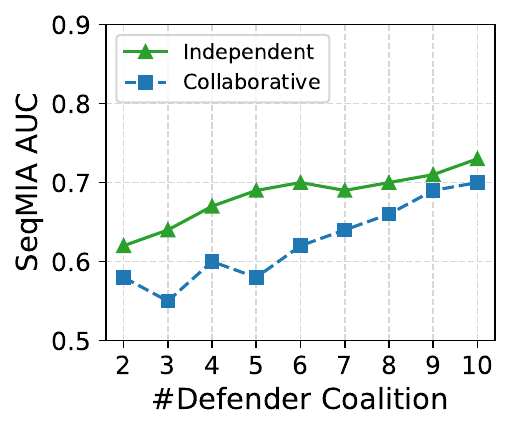}
		\label{fig:auc}
	}
	\subfigure{
		\centering
		\includegraphics[width=0.20\textwidth]{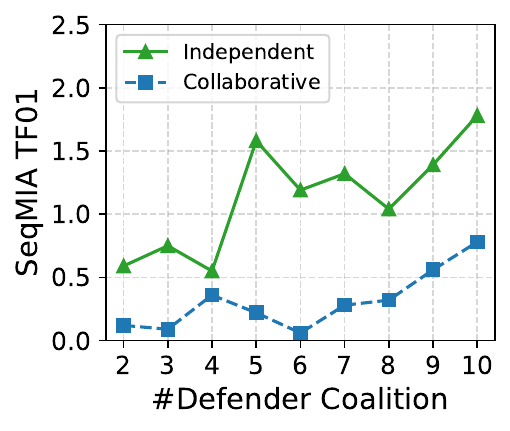}
		\label{fig:tf01}
	}
	\caption{Performance of various coalitions in both independent and collaborative modes on CIFAR100.}
	\label{fig:coll-indepent}
\end{figure}
\noindent \textbf{Performance in Different Defense Modes.}
To demonstrate the advantages of our proposed collaborative defense mode, we conduct experiments across a comprehensive range of defender coalitions—from pairs of clients up to full participation (i.e., 2 to 10 clients). The defense performance of both independent and collaborative modes under different coalition sizes is shown in Figure~\ref{fig:coll-indepent}.
Overall, the collaborative mode consistently improves the performance of our framework on both attack metrics, with TF01 exhibiting notably lower values than in the independent setting.

\noindent \textbf{Performance across Varying Total Clients.}
The preceding experiments were conducted on an FL system with 10 clients. Next, we assess the defense performance in FL systems with varying numbers of clients, ranging from 5 to 30, under both the Pair and Half cases.
As shown in Figure~\ref{fig:fl-system}, CoFedMID benefits from larger FL systems in the Pair case, achieving effective defense with lower utility degradation as more clients train. 
Furthermore, in the Half setting, increasing the number of clients does not have a significant impact on defense effectiveness or utility loss, resulting in results comparable to those in the 10-client scenario.
Thus, these results demonstrate the scalability of our framework in FL systems of varying sizes.

\begin{table}[t]
    \setlength{\abovecaptionskip}{0pt}  % 表格标题上方间距
    \setlength{\belowcaptionskip}{0pt}  % 表格标题下方间距
    \centering
    \footnotesize
    \caption{Performance under non-IID settings on CIFAR100 (AUC/TF01).}
    \label{tab:non_iid_c100}
    \begin{tabular}{lcccc}
    \toprule
    \multirow{2}{*}{\(\beta\)}  & \multicolumn{2}{c}{FedMIA-II} & \multicolumn{2}{c}{SeqMIA} \\
    \cmidrule(lr){2-3} \cmidrule(lr){4-5}
    ~& No Defense & CoFedMID & No Defense & CoFedMID \\
    \midrule
    $\infty$  & 0.83/12.64 & 0.56/3.04 & 0.92/11.80 & 0.50/0.00 \\
    10.0      & 0.85/26.51 & 0.58/3.53 & 0.97/32.93 & 0.54/0.03 \\
    1.0       & 0.91/58.57& 0.64/2.30 & 0.98/62.01 & 0.52/0.40 \\
    0.5       & 0.96/76.01 & 0.66/0.21 & 0.99/63.77 & 0.58/0.46 \\
    \bottomrule
    \end{tabular}
\end{table}

\begin{figure}[t]
    \setlength{\subfigcapskip}{0pt} % 图题与图片的距离
    \setlength{\subfigtopskip}{0pt} % 上方空白
    \setlength{\subfigbottomskip}{0pt} % 下方空白
	\centering
	\subfigure{
		\centering
		\includegraphics[width=0.23\textwidth]{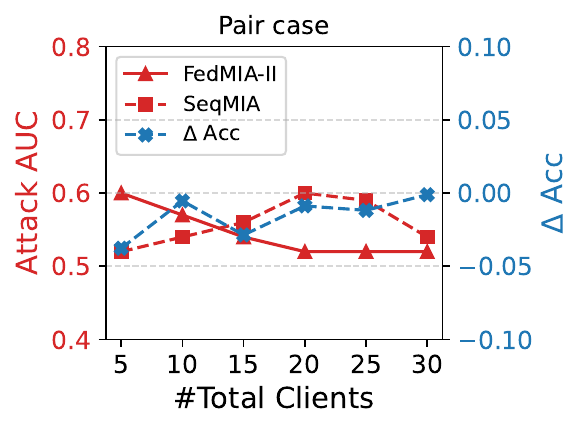}
	}\hspace{-3mm}
	\subfigure{
		\centering
		\includegraphics[width=0.23\textwidth]{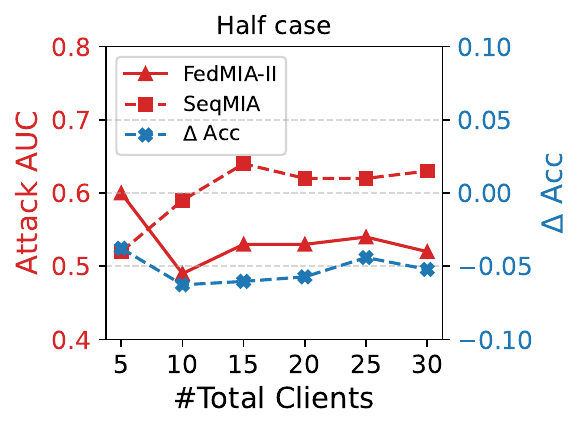}
	}
	\caption{Performance of various total clients on CIFAR100.}
	\label{fig:fl-system}
\end{figure}

\noindent \textbf{Performance in Non-IID settings.} 
We finally evaluate defense performance under different non-IID settings.
%where client data distributions are generated from a Dirichlet distribution with parameter $\beta$: smaller values correspond to higher heterogeneity and more skewed distributions.
To simulate these scenarios, we follow~\cite{HTH+19} and sample client data distributions from a Dirichlet distribution with parameter $\beta \in \{\infty, 10, 1.0, 0.5\}$, where smaller $\beta$ values indicate more skewed distributions and $\infty$ represents the IID setting.
We evaluate CoFedMID under various non-IID settings in the Pair case on CIFAR100, and present the results in Table~\ref{tab:non_iid_c100}.
Our findings show that CoFedMID consistently reduces attack performance across varying degrees of data heterogeneity.
While its effectiveness diminishes slightly under highly skewed distributions (e.g., $\beta = 0.5$) as attacks benefit from increased skewness, it still yields substantially lower AUCs and TF01s than the undefended case.
We present additional results for CIFAR10 and a broader range of attacks in Tables~\ref{tab:noniid-c100} and~\ref{tab:noniid-c10} in the Appendix.

\subsection{Detailed Analysis of CoFedMID}
\label{subsec:details}
In this subsection, we present a detailed analysis of the three CoFedMID modules under different configurations on CIFAR100 and CIFAR10 datasets, with half of the clients forming the defender coalition.

\begin{table}[htbp]
    \setlength{\abovecaptionskip}{0pt}  % 表格标题上方间距
    \setlength{\belowcaptionskip}{0pt}  % 表格标题下方间距
    \centering
    \footnotesize
    \caption{Impact of each module on defense performance (AUC/TF01) and model utility.   
    \ding{172}: class-guided partition module; \ding{173}: utility-aware compensation module; \ding{174}: aggregation-neural perturbation module.}
    \setlength{\tabcolsep}{3pt} % 调整列间距为5pt，默认是6pts
    \label{tab:ab_module}
    \begin{tabular}{l|l|llll}
        \toprule
        Dataset & MIA\&Acc & No defense & \multicolumn{1}{c}{\ding{172}} & \multicolumn{1}{c}{\ding{172}+\ding{173}} & \ding{172}+\ding{173}+\ding{174}  \\
        \midrule
        \multirow{3}{*}{CIFAR100}   
        & FedMIA-II & 0.82/12.36 & 0.51/0.14 & 0.52/0.23 & 0.51/0.48 \\
        & SeqMIA    & 0.89/10.62 & 0.59/0.57 & 0.64/0.73 & 0.59/0.71 \\ % 
        & $\Delta$Acc  &  \multicolumn{1}{c}{-} & -0.04    & -0.02   & -0.02    \\
        \midrule
        \multirow{3}{*}{CIFAR10}    
        & FedMIA-II & 0.62/1.17  & 0.53/0.00 & 0.55/0.42 & 0.54/0.77 \\
        & SeqMIA    & 0.88/7.34  & 0.54/0.16 & 0.64/0.96 & 0.60/1.74 \\
        & $\Delta$Acc  &  \multicolumn{1}{c}{-}   & -0.07    & -0.03    & -0.03    \\
        % \midrule
        % \multirow{3}{*}{TinyImageNet}  
        % & FedMIA-II & 0.83/15.66       & 0.57/0.74       & 0.58/0.68     & 0.53/0.91      \\
        % & SeqMIA    & 0.98/21.49       & 0.76/3.32       & 0.76/2.65      & 0.60/0.20     \\
        % & $\Delta$Acc  & \multicolumn{1}{c}{-} & -0.02  & -0.01   & -0.01      \\
        \bottomrule
    \end{tabular}
\end{table}
\noindent\textbf{Impact of Each Module.}
We first evaluate the individual contributions of the three CoFedMID modules to both defense performance and FL utility, with results for two representative attacks shown in Table~\ref{tab:ab_module}.
The results for the class-guided partition module (\ding{172}) show that reducing the exposure of local training samples effectively decreases membership privacy risk compared to undefended FL, but this comes at the cost of noticeable utility loss.
As a utility compensation, the module \ding{173} effectively mitigates this loss, particularly for CIFAR10. This difference arises from the data complexity: CIFAR10 holds a simpler data distribution, making the model more sensitive to additional training samples.
Moreover, while combining \ding{172} and \ding{173} is effective against FedMIA-II, incorporating the aggregation-level perturbation module (\ding{174}) further mitigates SeqMIA that exclusively targets local models.
Overall, these three modules contribute in complementary ways to balancing privacy protection against MIAs and model utility.

\begin{table}[t]
    \setlength{\abovecaptionskip}{0pt}  % 表格标题上方间距
    \setlength{\belowcaptionskip}{0pt}  % 表格标题下方间距
    \centering
    \footnotesize
    \caption{Impact of bounded class assignment (AUC/TF01). } %andom and bounded
    \label{tab:randomcls}
    \begin{tabular}{lcc|cc}
        \toprule
        \multirow{2}{*}{MIA\&Acc} & \multicolumn{2}{c|}{{CIFAR100}} & \multicolumn{2}{c}{{CIFAR10}} \\
        \cmidrule(lr){2-3} \cmidrule(lr){4-5}
                         & {Random} & {Bounded} & {Random} & {Bounded} \\
        \midrule
        FedMIA-II   & 0.53/0.19 & 0.51/0.14 & 0.52/0.02 & 0.53/0.00\\
        SeqMIA      & 0.66/1.31 & 0.59/0.57 & 0.53/0.18 & 0.54/0.06\\
        $\Delta$Acc    & -0.06 &  -0.04 & -0.11 & -0.07 \\ % use min 
        \bottomrule
    \end{tabular}
\end{table} 

\begin{table}[t]
    \setlength{\abovecaptionskip}{0pt}  % 表格标题上方间距
    \setlength{\belowcaptionskip}{0pt}  % 表格标题下方间距
    \centering
    \footnotesize
    \caption{Impact of decay-based class assignment (AUC). }
    \setlength{\tabcolsep}{1.5pt}
    \label{tab:staticcls}
    \begin{tabular}{lcccc|cccc}
    \toprule
    \multirow{2}{*}{MIA\&Acc} & \multicolumn{4}{c|}{CIFAR100} & \multicolumn{4}{c}{CIFAR10} \\
    \cmidrule(lr){2-5} \cmidrule(lr){6-9}
    % & F20 & F40 & F50 & Decay & F2 & F5 & F10 & Decay \\ % F8 -> F10
    & $\lambda_{min}$ & $\frac{\lambda_{min}+\lambda_{miax}}{2}$ & $\lambda_{max}$ & Decay & $\lambda_{min}$ & $\frac{\lambda_{min}+\lambda_{miax}}{2}$ & $\lambda_{max}$ & Decay \\ 
    \midrule
    FedMIA-II  & 0.51 & 0.54 & 0.56 & 0.49 & 0.48 & 0.50 & 0.54 & 0.54 \\
    SeqMIA & 0.62 & 0.67 & 0.73 & 0.59 & 0.54 & 0.66 & 0.68 & 0.54 \\
   $\Delta$Acc & -0.06 & -0.01 & -0.01 & -0.03 & -0.06 & -0.05 & -0.03 & -0.03 \\
    \bottomrule
    \end{tabular}
\end{table}

Next, we conduct an ablation study on various configurations within each module and analyze the effectiveness of their respective sub-modules.

\noindent\textbf{Analysis of Label-guided Partition Module.}
Recall that this module is designed to assign class subsets to clients with minimal overlap, following a bounded, decay-based strategy. To assess its effectiveness, we implement it in isolation and evaluate the impact of different configurations.

\textit{1. Random vs. Bounded Class Assignment.} 
We compare our bounded assignment strategy with random sampling, where each client is assigned a class subset selected uniformly at random.
As shown in Table~\ref{tab:randomcls}, the bounded strategy achieves comparable defense against MIAs while reducing utility loss, indicating that controlling inter-client class overlap preserves the overall class distribution and helps maintain FL performance.

\textit{2. Constant vs. Decay-based Class Assignment.}
We evaluate our decay-based strategy against a constant approach, where the number of assigned classes remains fixed throughout training. 
Table~\ref{tab:staticcls} shows that fixed class assignments suffer from a dilemma: assigning fewer classes generally improves defense but causes greater utility degradation.
In contrast, our strategy achieves defense performance similar to that of $\lambda_{\min}$ and utility close to that of $\lambda_{\max}$.

Additionally, we compare linear and non-linear decay strategies and find that non-linear approaches, such as the exponential function, can further reduce the attack performance, as detailed in Appendix~\ref{app:ablation}.

\begin{table}[t]
    \setlength{\abovecaptionskip}{0pt}  % 表格标题上方间距
    \setlength{\belowcaptionskip}{0pt}  % 表格标题下方间距
    \centering
    \footnotesize
    \caption{Comparison of different recycling strategies (AUC). }
    \setlength{\tabcolsep}{4pt}
    \label{tab:sampling}
    \begin{tabular}{lccc|ccc}
    \toprule
    \multirow{2}{*}{MIA\&Acc} & \multicolumn{3}{c|}{CIFAR100} & \multicolumn{3}{c}{CIFAR10} \\
    \cmidrule(lr){2-4} \cmidrule(lr){5-7}
    & Rand & Seq & MAB & Rand & Seq & MAB  \\
    \midrule
    FedMIA-II       & 0.52   & 0.52 & 0.52  & 0.55 & 0.56 & 0.55  \\
    SeqMIA          & 0.64   & 0.65 & 0.63  & 0.66 & 0.67 & 0.62 \\
   $\Delta$Acc & -0.03   & -0.03 & -0.03  & -0.08 & -0.09 & -0.04 \\
    \bottomrule
    \end{tabular}
\end{table}

\begin{table}[t]
    \setlength{\abovecaptionskip}{0pt}  % 表格标题上方间距
    \setlength{\belowcaptionskip}{0pt}  % 表格标题下方间距
    \centering
    \footnotesize
    \caption{Impact of confidence regularization (AUC).}
    \label{tab:ab_hamp}
    \begin{tabular}{lcc|cc}
    \toprule
    \multirow{2}{*}{MIA} & \multicolumn{2}{c|}{CIFAR100} & \multicolumn{2}{c}{CIFAR10} \\
    \cmidrule(lr){2-3} \cmidrule(lr){4-5}
                    & w/o CR & w/CR & w/o CR & w/CR  \\
    \midrule
    FedMIA-II       & 0.55   & 0.52 & 0.56  & 0.52   \\ % auc
    SeqMIA          & 0.80   & 0.59 & 0.73  & 0.61  \\ % same as Table 6
    \bottomrule
    \end{tabular}
\end{table}

\begin{figure}[t]
    \setlength{\subfigcapskip}{0pt} % 图题与图片的距离
    \setlength{\subfigtopskip}{0pt} % 上方空白
    \setlength{\subfigbottomskip}{0pt} % 下方空白
	\centering
	\subfigure{
		\centering
		\includegraphics[width=0.23\textwidth]{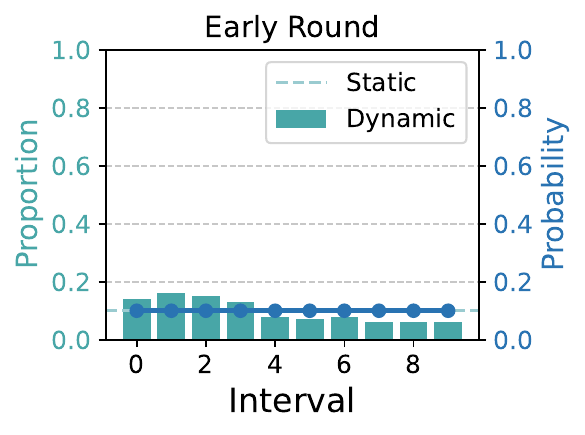}
	}\hspace{-3mm}
	\subfigure{
		\centering
		\includegraphics[width=0.23\textwidth]{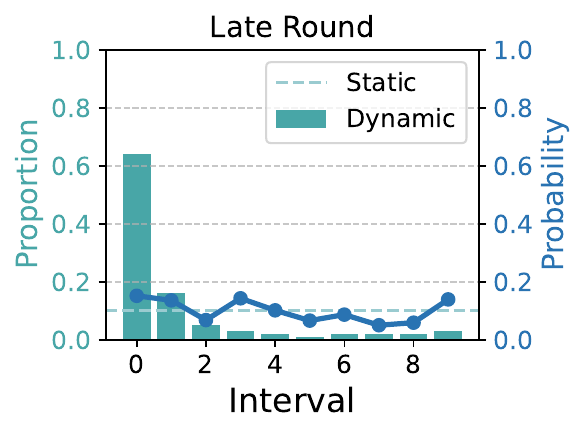}
	}\hspace{-3mm}
	\caption{Proportions and probabilities of sample intervals in early and late training rounds.}
	\label{fig:mabpro}
\end{figure}
\noindent\textbf{Analysis of Utility-aware Compensation Module.}
The utility-aware compensation module reduces the utility loss from the class-guided partition module by selectively reintroducing a small subset of high-contribution samples. 
We evaluate its effectiveness against two heuristic baselines, examine the role of confidence regularization on these samples, and provide additional results for varying recycling ratios in Appendix~\ref{app:ablation}.

\textit{1. Comparison of Recycling Strategies.}
We compare our MAB-based sample recycling strategy (\textit{MAB}) with two heuristic approaches. The random strategy (\textit{Rand}) uniformly selects a fixed number of samples from the remaining local dataset, while the sequential strategy (\textit{Seq}) reintroduces samples in order of their loss values, from low to high. Both baselines recycle a fixed ratio (i.e., 10\%) of samples from local data in each round. 

We present the comparison results of the three approaches in Table~\ref{tab:sampling}. While CIFAR100 exhibits little sensitivity to the choice of recycling strategy due to its diverse data distribution, our method achieves both lower attack performance on both datasets and reduced utility loss on CIFAR10.
To further demonstrate the superiority of the MAB-based recycling strategy, we visualize the sampling scale over training in Figure~\ref{fig:mabpro}, showing how it adaptively selects contributive samples in both early and later rounds. 
Compared to random and sequential methods (i.e., Static) that employ a fixed sampling size, our approach (i.e., Dynamic) consistently recycles fewer high-contribution samples, thereby reducing unnecessary exposure of local data and strengthening membership privacy protection.

\textit{2. Impact of Confidence Regularization.}
We conduct an ablation study to evaluate the effectiveness of the confidence regularization (CR) approach, which is designed to prevent membership privacy leakage from recycled samples.
As shown in Table~\ref{tab:ab_hamp}, incorporating CR substantially reduces MIA success on both datasets, with SeqMIA showing roughly an order-of-magnitude drop, demonstrating its effectiveness in mitigating membership privacy risks.

\noindent\textbf{Analysis of Aggregation-neutral Perturbation Module.}
We have shown the effectiveness of the aggregation-level perturbation module in defending against MIAs targeting local models (see Table~\ref{tab:ab_module}). We further examine the impact of its hyperparameters, namely perturbation strength and ratio. 
As illustrated in Figure~\ref{fig:perturb} for CIFAR100, the perturbation strength has a notable impact on defense performance, with the AUC decreasing by approximately 0.05 as the strength increases from 0.05 to 0.10. 
In the case of the perturbation ratio, its impact on defense effectiveness is relatively minor. For instance, the AUC changes by about 0.03 when the ratio increases from 0.10 to 0.25. 
As a practical guideline, it is preferable to focus on an appropriate perturbation strength and avoid an excessively small perturbation ratio.

\subsection{Adaptive Attacks}
\label{subsec:adap}
An effective defense should demonstrate robustness against adaptive adversaries, who possess complete knowledge of the defense and can modify their attack strategies in response, especially when the existence of the coordinator or the coalition may be exposed. 
To evaluate the robustness of CoFedMID under this threat model, we design an adaptive attack tailored to challenge our defense mechanism.
We assume the attacker can identify clients within the defender coalition, using auxiliary information such as device IPs, public institutional affiliations, or other metadata that reveal client identities or associations. 
Rather than performing MIAs on the local or global models of individual clients, the attacker directs its efforts toward \textit{the aggregated model of the defender coalition}, considering this collective model as the attack target.

To evaluate adaptive attacks, we adopt four different attack methods to this threat model, replacing their respective target models with the aggregated model for each attack.
We conduct experiments on both CIFAR10 and CIFAR100 datasets, with the resulting defense performance summarized in Table~\ref{tab:adaptive}.
We find that the aggregated model results in a noticeable reduction in attack performance compared to non-adaptive attacks in undefended settings. With respect to our framework, CoFedMID, it demonstrates strong defense capabilities against adaptive attacks.
For the first three MIAs, CoFedMID maintains defense performance comparable to that observed under the non-adaptive attacks. Although SeqMIA exhibits an increased AUC under adaptive attacks as the aggregation-level perturbation module is bypassed, its TF01 remains below 2\%, indicating a poor ability to distinguish between member and non-member samples. Overall, these results demonstrate that our defense framework is robust against adaptive attacks.
  
\begin{table}[t]
    \setlength{\abovecaptionskip}{0pt}
    \setlength{\belowcaptionskip}{0pt}
    \centering
    \footnotesize
    \caption{Defense performance against adaptive attacks.}
    \setlength{\tabcolsep}{4pt}
    \label{tab:adaptive}
    \begin{tabular}{l|cccc}
    \toprule
    Case & Loss-Series & FedMIA-II & FTA-C & SeqMIA \\
    \midrule
    \multicolumn{5}{c}{CIFAR100 Dataset} \\ \midrule
    \addlinespace[1pt]
    No defense         & 0.66/13.3  & 0.68/1.48 & 0.60/3.26 & 0.70/0.00 \\
    CoFedMID (Pair)    & 0.46/2.78  & 0.56/5.47 & 0.49/1.14 & 0.52/0.00 \\
    CoFedMID (Half)    & 0.50/3.08  & 0.55/5.36 & 0.45/1.80 & 0.66/0.00 \\
    \midrule
    \addlinespace[3pt]
    \multicolumn{5}{c}{CIFAR10 Dataset} \\ \midrule
    \addlinespace[1pt]
    No defense         & 0.62/11.1  & 0.55/0.29 & 0.50/1.70 & 0.80/0.03 \\
    CoFedMID (Pair)    & 0.54/3.04  & 0.52/0.96 & 0.53/0.80 & 0.63/1.98 \\
    CoFedMID (Half)    & 0.54/2.20  & 0.55/0.38 & 0.50/0.16 & 0.64/0.32 \\
    \bottomrule
    \end{tabular}
\end{table}

\begin{figure}[t]
    \setlength{\subfigcapskip}{0pt} % 图题与图片的距离
    \setlength{\subfigtopskip}{0pt} % 上方空白
    \setlength{\subfigbottomskip}{0pt} % 下方空白
    \centering
    \includegraphics[width=0.46\textwidth]{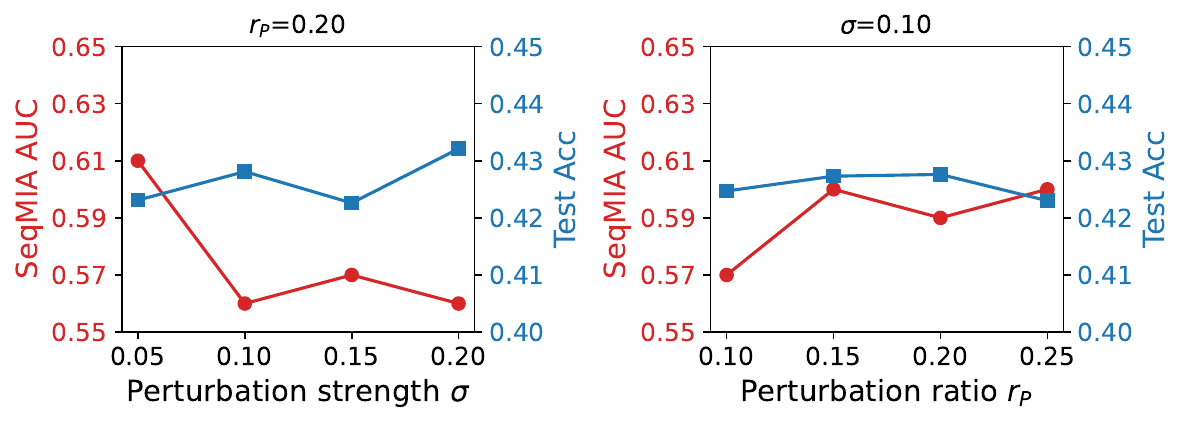}
    \caption{Impact of perturbation strength and ratio.} %on CIFAR100
    \label{fig:perturb}
\end{figure}

\subsection{Overhead of CoFedMID}
We assess the overhead of our defense framework for FL, including communication and computational costs.

\noindent\textbf{Communication Cost.}
The communication cost is a critical concern in FL, particularly for limited network bandwidth.
Our framework incurs extra communication overhead before training, as a leader client distributes class assignments and perturbation configurations to the coalition.
Specifically, each client receives up to $N$ class labels along with two scalar values for perturbation configurations per training round. Given $T$ rounds, this totals $(N+2) \times T$ items. Assuming each item occupies 1 byte, the total communication overhead is approximately 20~KB for our default settings (e.g., $T=100$ and $N=200$ for TinyImageNet).
Compared with model updates exchanged in FL, which are typically tens of megabytes per round, CoFedMID adds only minimal communication overhead while offering effective defense against MIAs.

\noindent\textbf{Computational Cost.} 
Effective defenses should not only provide strong protection but also remain efficient to deploy in practice.
We evaluate the computational cost of our defense by measuring the client-side local training time across the entire FL process, as summarized in Table~\ref{tab:runtime}.
As expected, the class-guided partition module (\ding{172}) reduces runtime compared to undefended FL by using fewer training samples, the perturbation module (\ding{174}) adds negligible overhead due to its efficient design, while the compensation module (\ding{173}) incurs additional computational cost.
Overall, the results show that CoFedMID incurs only about 1.4× the runtime of undefended FL on average, which remains acceptable for clients.
\begin{table}[t]
    \setlength{\abovecaptionskip}{0pt}  % 表格标题上方间距
    \setlength{\belowcaptionskip}{0pt}  % 表格标题下方间距
    \centering
    \footnotesize
    \caption{Runtime of local training over (seconds).}
    \setlength{\tabcolsep}{7pt}
    \label{tab:runtime}
    \begin{tabular}{l|lccc}
        \toprule
        Dataset & \multicolumn{1}{c}{No Defense} &  \ding{172} & \ding{172}+\ding{173} & \ding{172}+\ding{173}+\ding{174} \\
        \midrule
        CIFAR100   & 1x (225.3) & 0.78x & 1.53x & 1.52x \\
        CIFAR10    & 1x (229.8) & 0.87x & 1.66x & 1.69x \\
        TinyImageNet  & 1x (1003.2) & 0.57x & 0.98x & 1.03x \\
        \midrule
        Average & \multicolumn{1}{c}{-} & 0.73x & 1.39x & 1.41x \\
        \bottomrule
    \end{tabular}
\end{table}

\section{Discussion}
\label{sec:dis}

\textbf{The Reason Why CoFedMID Works.}
% \subsection{The Reason Why CoFedMID Works}
The above experiments demonstrate that our proposed framework is effective in defending against MIAs in FL both in uniform and partial scenarios.
To further understand its effectiveness, we now analyze the underlying mechanisms of CoFedMID.
We first show the average loss trajectories of a client's local model within the defender coalition in Figure~\ref{fig:our-loss}.
In conjunction with the loss trajectory for undefended FL in Figure~\ref{fig:noneloss}, our framework elevates the overall sample loss and significantly reduces the loss gap between member samples and both types of non-member samples, thereby substantially lowering the attack performance in practice.

Moreover, our framework is designed to mitigate MIAs in FL by reducing the exposure of training samples.
To investigate this, we randomly select a subset of training samples from CIFAR100 and record their usage in the local training process.
As shown in Figure~\ref{fig:our-data}, most samples are trained only about half as many times as in the undefended setting (Used). 
In addition, each sample is protected by the additional defense mechanism over several rounds (Used (w/ CR)). This not only prevents a significant discrepancy from the training process of undefended clients but also narrows the distinguishability between member and non-member IFL samples.
Furthermore, it indicates that our framework employs distinct strategies for different training samples. For instance, it allocates stronger protection to high-loss samples while providing less additional safeguarding to low-loss samples. This differentiated way enhances the overall effectiveness of our defense.

\begin{figure}[t]
    \setlength{\subfigcapskip}{0pt} % 图题与图片的距离
    \setlength{\subfigtopskip}{0pt} % 上方空白
    \setlength{\subfigbottomskip}{0pt} % 下方空白
	\centering
	\subfigure[Loss trajectory]{
		\centering
		\includegraphics[width=0.215\textwidth]{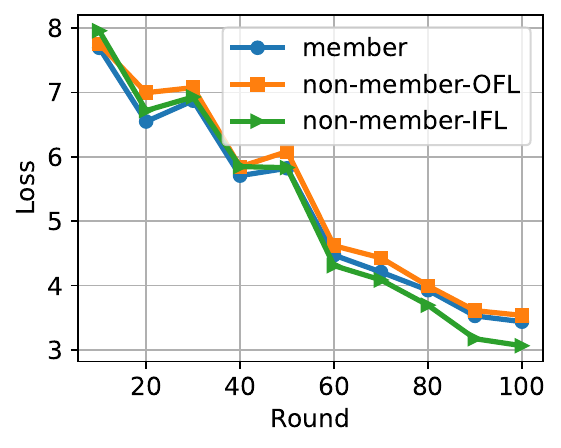}
        \label{fig:our-loss}
	}
	\subfigure[Data usage and protection]{
		\centering
		\includegraphics[width=0.23\textwidth]{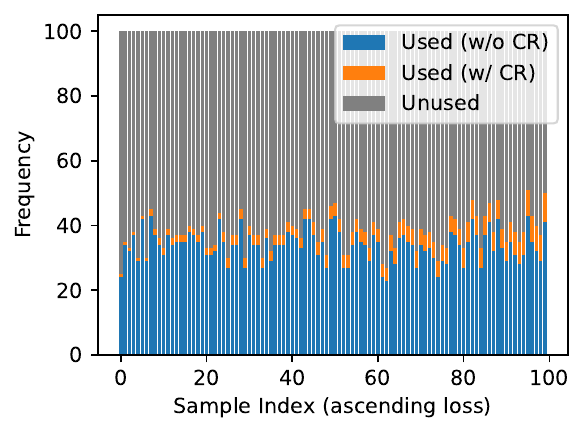}
        \label{fig:our-data}
	}
	\caption{Analysis of CoFedMID’s effectiveness.}
%	\label{fig:mabpro}
\end{figure}

\noindent\textbf{Limitations and Future Work.}
% \subsection{Limitations and Future Work}
Despite the demonstrated superiority, CoFedMID has two limitations that warrant further investigation:

\textit{Privacy-Utility Trade-offs.} 
Our framework cannot fully avoid the inherent trade-off between privacy and utility, particularly in large coalitions. As more clients participate, a greater portion of training data must be concealed, a more pronounced impact on model performance. 
Future research could mitigate this trade-off by adaptively determining the initial class subset size in the assignment strategy based on the coalition size.

\textit{Poisoning MIAs.}
Advanced MIAs combined with poisoning attacks represent an emerging research direction~\cite{MWL+25}. 
In this work, we focus on MIAs initiated by honest-but-curious clients or servers, and develop our defense framework under the assumption of an honest coordinator and coalition members.
However, the distributed nature of FL makes it susceptible to poisoning attacks. 
Since defending against poisoning attacks is an orthogonal problem to membership inference defense, we leave mitigating these complex attack scenarios to future work.

\textit{Task Extension.}
Our framework requires explicit label information to construct the class-guided module, which makes it naturally effective for classification tasks. Consequently, it does not directly apply to label-free settings such as unsupervised learning. Exploring extensions such as pseudo-labeling~\cite{LDO+13} is a promising direction for future work. 
Additionally, while the majority of MIA studies have concentrated on classification tasks—particularly those involving image data—future work could broaden the scope to include other data modalities, such as text and medical records.

% \textbf{Leader of Defender Coalition.}
% We implicitly assume that a leader client is responsible for distributing the class assignments and perturbation configurations in a collaborative setting. As this task does not require significant computational resources, we consider that any randomly selected client can perform this preparation.
% While this intuitive approach may introduce potential problems (e.g., uncertainty), these can be mitigated by adopting deterministic election principles, which we leave for future work.

% \textit{Adaptive Attacks.}
% Although Section~\ref{subsec:adap} shows that adaptive attacks cannot significantly break through our framework, if adversaries identify the defender coalition, our defense may become vulnerable to breach.
% This issue can be addressed by introducing dummy clients to obscure the mapping between real clients and their updates, which represents a potential direction for future work.
\section{Related Works}
\label{sec:rel}

\subsection{Membership Inference Attacks}
% \noindent\textbf{Membership Inference Attacks.}
MIAs are first introduced by Shokri et al.~\cite{SRS+17} in the context of centralized learning to determine whether a specific data record belonged to the training set of a target model. 
% Their method trains multiple shadow models to simulate the target model’s behavior, followed by a meta-classifier that distinguishes members from non-members based on model outputs. 
Since then, numerous approaches have been proposed, substantially enhancing inference performance across diverse model architectures and settings~\cite{CDY+20,CNC+22,WYY+22,WRL+24}.

Building upon these centralized works, research has extended MIAs to federated settings. Nasr et al.~\cite{NMS+19} first proposed MIAs in FL by exploiting per-sample gradients, intermediate features, and loss values during training. This work reveals significant privacy risks in model updates when local data is insufficiently protected. Li et al.~\cite{LJL+23} later developed a gradient-based attack leveraging cosine similarity between per-sample gradients across communication rounds.

Compared to the above update-based MIAs, trajectory-based approaches demonstrate stronger performance.
These attacks capture richer temporal information by exploiting how certain statistics evolve throughout the training process. For instance, an attacker may monitor the trajectory of per-sample loss values, prediction confidence~\cite{GYB+22}, or gradient norms~\cite{ZGL+25} to distinguish member samples from non-members.
FedMIA~\cite{ZGL+25} adapts LiRA~\cite{CNC+22} for FL by using aggregated signals like loss values and gradient norms over multiple rounds to perform membership inference.
Similarly, Chen et al.~\cite{CHE+24} propose a lightweight yet effective method that leverages the rate of change of per-sample loss or prediction confidence over multiple model snapshots, achieving strong attack performance.

% \noindent\textbf{Membership Inference Defenses.}
\subsection{Membership Inference Defenses}
A variety of defenses against MIAs have been extensively studied in both centralized and federated contexts. 
Centralized defenses include output perturbation to reduce confidence gaps~\cite{SRS+17, JJS+19, ZCK+24}, regularization techniques like data augmentation and adversarial training to mitigate overfitting~\cite{KYD+21, NMS+18}, knowledge distillation to transfer knowledge to protected models~\cite{SVH+21, TXM+22}, and differential privacy to provide formal privacy guarantees~\cite{TSL+19, JBE+19, AMZ+24}.
These methods mainly reduce the difference between member and non-member data by preventing the target model from overfitting.

However, due to the decentralized nature of FL, defensive mechanisms differ substantially, typically relying on perturbation-based approaches.
Partial sharing reduces the attack surface by selectively suppressing or filtering model updates during training, such as gradient compression~\cite{MLS+19, LJL+23} and weight pruning~\cite{SDG+22}, to limit exposure of sensitive information. Nevertheless, this attack-agnostic approach provides limited defense performance~\cite{MLS+19}.
Differential privacy remains a key defense in FL by adding calibrated noise to gradients to obscure individual contributions and reduce membership leakage. The effectiveness of local and central DP against white-box attacks is demonstrated in~\cite{MNJ+22}, often at the cost of model utility~\cite{HHS+21, SAK+22}. The work~\cite{YXF+22} offers effective defense via model perturbation but requires a trusted server.
An orthogonal line of defense employs cryptographic techniques~\cite{BKI+17, AYH+17} to protect membership privacy. For instance, SecAgg~\cite{BKI+17} uses secret sharing and masking to defend against honest-but-curious servers. While these methods secure model updates during transmission, they incur substantial computational and communication overhead. 
Furthermore, they cannot prevent MIAs via the global model, as clients still have access to the decrypted global model~\cite{GYB+22}.
% However, it cannot prevent attacks from malicious clients who access the global model.
% While these methods secure model updates during transmission, they incur high computational and communication overhead and cannot prevent MIAs via global models, since either the server or clients can still access the decrypted model~\cite{GYB+22}.
\section{Conclusion}
\label{sec:con}

In this work, we explore a new and realistic scenario in which a group of clients collaborates to form a defender coalition against MIAs in FL. 
Towards this end, we propose CoFedMID, a novel defense framework designed to mitigate membership privacy leakage in such settings. 
Extensive experiments show that our approach consistently surpasses existing methods from both federated and centralized settings. 
In addition to our primary scenario, CoFedMID also outperforms existing defenses in the previously studied uniform and independent scenarios.
We underscore the value of client collaboration for defense and hope that our work will inspire further investigation into advanced mitigation strategies in FL.
% \newpage

% %-------------------------------------------------------------------------------
% \section*{Acknowledgments}
% %-------------------------------------------------------------------------------
% The USENIX latex style is old and very tired, which is why
% there's no \textbackslash{}acks command for you to use when
% acknowledging. Sorry.

% %-------------------------------------------------------------------------------
% \section*{Availability}
% %-------------------------------------------------------------------------------
% USENIX program committees give extra points to submissions that are
% backed by artifacts that are publicly available. If you made your code
% or data available, it's worth mentioning this fact in a dedicated
% section.

\section*{Ethical Considerations}
In this study, we propose an effective defense against membership inference attacks in federated learning and analyze its potential impacts. To ensure ethical research practices, we discuss the following points.

\textbf{Safety Considerations for Experiments.} All experiments were conducted exclusively on publicly available benchmark datasets within a simulated federated learning environment. These datasets contain no personally identifiable or sensitive information, and all clients and servers were instantiated in a controlled virtual setting, without involving any real users, devices, or organizations.

\textbf{Positive Impacts of Proposed Defenses.} Our defense provides protection for clients’ local data against membership inference attacks. By mitigating the risk of privacy leakage to honest-but-curious servers and clients outside the coalition, the defense primarily benefits FL clients, who are the main recipients of its protective effects. In real-world scenarios, industries such as healthcare, finance, and mobile applications that rely on federated learning could see improved privacy for users’ sensitive data if proposed defenses are deployed.

\textbf{Negative Impacts of Proposed Defenses.} Despite these benefits, there are potential negative consequences that must be pointed out: (1) Protected clients may overestimate the defense’s effectiveness if coalition members misbehave and collude. (2) While the defense offers partial protection for sensitive users, it does not cover clients outside the coalition, leaving their data potentially vulnerable to attackers. (3) Attackers may adapt strategies to bypass the defense or exploit insights from it. (4) Behavioral shifts among clients may occur, as organizations and users might over-rely on the protection mechanism, potentially altering the overall risk profile and unintentionally amplifying the effectiveness of attacks.

\section*{Open Science}
To support reproducibility and replicability, we provide the research artifacts of our defense framework, including source code, scripts, and detailed instructions, to facilitate the reproduction of our experimental results. They are available at \href{https://zenodo.org/records/17982729}{Zenodo repository}, which include:

\begin{itemize}
	\item \textbf{Datasets:} We use three publicly available datasets (i.e., CIFAR10, CIFAR100, TinyImageNet) and provide instructions for downloading them.
	\item \textbf{Source Code:} We provide the core implementation of the proposed defense method as well as baseline algorithms.
	\item \textbf{Scripts:} We include two scripts for running experiments on CIFAR100 and ResNet18 to evaluate the defense performance of CoFedMID and the baselines.
	\item \textbf{Evaluation Demo:} We provide a Jupyter notebook for statistical analysis and visualization of the experimental results.
\end{itemize}

All artifacts, along with detailed instructions and full version, are also accessible at \href{https://github.com/BaiLibl/CoFedMID}{Github repository}.

\section*{Acknowledgments}
This work was supported by the National Natural Science Foundation of China (Grant No:  92270123, 62372122, 62502416, and U25A20430),  and the Research Grants Council (Grant No: 15208923, 25207224, and 15207725), Hong Kong SAR, China.

%---------------------------
\bibliographystyle{unsrt}
\bibliography{refer.bib}

\newpage
\appendix

\section{Proof Statement}
In this section, we present the proof of the theoretical lower bound on the overlap between fixed-size subsets.

\begin{theorem}
Given a set of $N$ classes and $K$ users (or clients), suppose each user is assigned a subset of $m$ classes. Let 
\[
S_1, S_2, \dots, S_K \subseteq \{1, \dots, N\}
\]
be the subsets assigned to the $K$ users, with $|S_i| = m$ for all $i = 1, \dots, K$. Then the maximum overlap between any pair of users satisfies the lower bound:
\[
\max_{1 \le i < j \le K} |S_i \cap S_j| \ge \left\lceil \frac{m(Km - N)}{N(K - 1)} \right\rceil.
\]
\end{theorem}

\begin{proof}
Define the total overlap over all pairs of users as
\[
T := \sum_{1 \le i < j \le K} |S_i \cap S_j|.
\]

For each class $c \in \{1, \dots, N\}$, let $f_c$ be the number of users whose subset contains class $c$. Since each user has $m$ classes, the total number of class assignments is
\[
\sum_{c=1}^N f_c = K m.
\]

Each class $c$ contributes $\binom{f_c}{2} = \frac{f_c(f_c - 1)}{2}$ to the total overlap $T$, as it is counted once for every pair of users sharing that class. Therefore,
\[
T = \sum_{c=1}^N \binom{f_c}{2} = \sum_{c=1}^N \frac{f_c (f_c - 1)}{2}.
\]

By Jensen's inequality and convexity of the function $x \mapsto x(x-1)$, the minimum of $T$ occurs when all $f_c$ are as equal as possible. Let the average frequency be
\[
\bar{f} = \frac{1}{N} \sum_{c=1}^N f_c = \frac{K m}{N}.
\]

Thus, 
\[
T \ge N \cdot \frac{\bar{f}(\bar{f} - 1)}{2} = \frac{K m (K m - N)}{2 N}.
\]

There are $\binom{K}{2} = \frac{K(K-1)}{2}$ user pairs in total, so by the pigeonhole principle, at least one pair has overlap at least
\[
\max_{i<j} |S_i \cap S_j| \ge \frac{T}{\binom{K}{2}} \ge \frac{m (K m - N)}{N (K - 1)}.
\]

Since the maximum overlap must be an integer, we take the ceiling:
\[
\max_{1 \le i < j \le K} |S_i \cap S_j| \ge \left\lceil \frac{m (K m - N)}{N (K - 1)} \right\rceil.
\]
\end{proof}

\begin{algorithm}
\caption{Class-guided Partition}
\label{alg:partition}
\KwIn{Number of classes $N$; coalition size $d$; maximum classes per client $m_{\max}$; minimum classes per client $m_{\min}$; total rounds $T$}
\KwOut{Subsets $\{\mathcal{S}_1^t, \dots, \mathcal{S}_d^t\}_{t=1}^T$}

\For{$t \gets 1$ \KwTo $T$}{
    $m^t \leftarrow$ \texttt{Decay}$(m_{\max}, m_{\min}, t)$\;
    Initialize overlap bound $\lambda \leftarrow \lambda_{theo}(m^t)$ \;
    \Repeat{success}{
        Initialize $\mathcal{S}_1^t, \dots, \mathcal{S}_d^t \leftarrow \emptyset$, success $\leftarrow$ \texttt{True}\;
        \For{$i \leftarrow 1$ \KwTo $d$}{
            \uIfNot{$ClassAssign(N, d, m^t, \lambda)$}{
                \tcp{Overlap constraint violated}
                success $\leftarrow$ \texttt{False}; $\lambda \leftarrow \lambda + 1$; \textbf{break}
            }
        }
    }
}
\Return{$\{\mathcal{S}_1^t, \dots, \mathcal{S}_d^t\}_{t=1}^T$}
\end{algorithm}

\begin{algorithm}[h]
\caption{Utility-Aware Compensation}
\label{alg:compensation}
\KwIn{
Local training set $D_k$, validation set $D_{\text{val}}$; \\
Total rounds $T$, number of intervals $M$; \\
Initialization round $t_0$, maximum recycling ratio $r_l$
}
\KwOut{Local model $\theta_k$}
Initialize interval weights $w_j \gets 1$ for all $j = 1, \dots, M$\;
Initialize sample intervals $\{\mathcal{I}_j\}_{j=1}^{M}$ at round $t_0$\;

\For(\tcp*[f]{Recycling begins}){$t \gets t_0 + 1$ \KwTo $T$}{
    Receive global model $\tilde{\theta}$ and update local model $\theta_k$\;
    Obtain $D_k^{\text{ass}}$ from the first module\;
    Sample interval index $j \sim \pi_w$\;
    \tcp{Collect recycled samples from $\mathcal{I}_j$}
    $D_k^{\text{rec}} \gets \{(x_i, y_i) \in D_k \setminus D_k^{\text{ass}} \mid (x_i, y_i) \in \mathcal{I}_j\}$\;
    $D_k^{\text{rec}} \gets \text{RandomSubset}(D_k^{\text{rec}}, \lfloor r_l \cdot |D_k| \rfloor)$\;

    \tcp{Local update}
    Train $\theta_k$ on $D_k^{\text{ass}} \cup D_k^{\text{rec}}$ using the combined loss\;

    \tcp{Reward estimation and EXP3 update}
    Evaluate validation losses of $\tilde{\theta}$ and $\theta_k$ on $D_{\text{val}}$\;
    Compute normalized reward $\tilde{r}_j$\;
    Update weight and sampling probability accordingly\;
}
\Return $\theta_k$
\end{algorithm}

\begin{algorithm}[h]
\caption{Aggregation-Neutral Perturbation}
\label{alg:perturb}
\KwIn{
  Model parameters \(\{\theta_k\}_{k=1}^{d}\) of defender coalition;
  Aggregation weights \(\{w_k\}_{k=1}^{d}\); 
  Perturbation strength \(\sigma\); Perturbation ratio \(r_p\)
}
\KwOut{Perturbed model parameters \(\{\theta_k\}_{k=1}^{d}\)}

\For{\(k = 1\) \KwTo \(d\)}{
Sample \(\delta_k' \sim \mathcal{N}(0, \sigma^2 I)\)\;
}
\For{\(k = 1\) \KwTo \(d\)}{
    Compute projection: \(\delta_k = \delta'_k - \frac{w_k}{\|\boldsymbol{w}\|_2^2} \cdot \sum_{j=1}^{d} w_j \delta'_j\)\;
  Let $\theta_k$ have \(P\) parameters and define \(P^{\prime} = \lfloor r_p \cdot P \rfloor\)\;
  \For{\(i = P - P^{\prime} + 1\) \KwTo \(P\)}{
    Inject perturbation: \(\theta_k^{(i)} \gets \theta_k^{(i)} + \delta_k\)\;
  }
}
\Return \(\{{\theta}_k\}_{k=1}^{d}\)

\end{algorithm}

\section{Algorithms}
\label{sec:app-alg}

In this section, we provide the pseudocode of the three core modules: 
(i) \textbf{Class-guided partition} (Algorithm~\ref{alg:partition}), which limits the exposure of individual datasets by distributing the overall data distribution across clients, each handling only a local subset; (ii) \textbf{Utility-aware compensation} (Algorithm~\ref{alg:compensation}), which mitigates utility loss by selectively increasing exposure in a performance-preserving manner; 
and (iii) \textbf{Aggregation-neutral perturbation} (Algorithm~\ref{alg:perturb}), which enhances defense by enabling defenders to collaboratively inject noises that cancel out during aggregation.
Furthermore, we present the pseudocode of the overall framework that integrates these modules in Algorithm~\ref{alg:overall}.

\begin{algorithm}[h]
\caption{Federated Learning with CoFedMID}
\label{alg:overall}
\KwIn{Number of clients $K$, defender coalition $\mathcal{D}$, total rounds $T$}
\KwOut{Global model $\tilde{\theta}$}
\tcp{Preparation by the coordinator}
Initialize class subsets $\{\mathcal{S}_1^t, \dots, \mathcal{S}_d^t\}_{t=1}^T$\;
Precompute perturbation noises $\{\delta_k\}_{k \in \mathcal{D}}$\;
\tcp{FL process}
\For{$t \gets 1$ \KwTo $T$}{
    Server samples a subset $\mathcal{K}$ from $K$ clients\;

    \tcp{Client-side local update}
    \ForEach{$k \in \mathcal{K}$}{
        Initialize local model $\theta_k \gets \tilde{\theta}$\;

        \eIf{$k \in \mathcal{D}$}{
            \tcp{Coalition client training}
            $D_k^{\text{ass}} \gets \texttt{Partition}(D_k, \mathcal{S}_k^t)$\;
            $D_k^{\text{rec}} \gets \texttt{Compensation}(D_k, D_k^{\text{ass}})$\;
            Train $\theta_k$ on $D_k^{\text{ass}} \cup D_k^{\text{rec}}$\;
            Perturb $\theta_k$ with $\delta_k$\;
        }{
            \tcp{Standard client training}
            Update $\theta_k$ on $D_k$ using $\ell_{ce}$\;
        }
        Send $\theta_k$ to server\;
    }

    \tcp{Server-side model aggregation}
    $\tilde{\theta} \gets \sum_{k \in \mathcal{K}} \frac{|D_k|}{\sum_{j \in \mathcal{K}} |D_j|} \cdot \theta_k$\;
}
\Return{$\tilde{\theta}$}
\end{algorithm}

\section{Dataset Details}
Considering class diversity, we conduct experiments on three widely used benchmark image classification datasets in MIA research: CIFAR10~\cite{KAH+09}, CIFAR100~\cite{KAH+09}, and TinyImageNet~\cite{YLX+15}. 

\noindent  \textit{ {1) CIFAR10}}~\cite{KAH+09} is a widely adopted benchmark in computer vision, consisting of 60,000 color images of size 32×32 pixels, evenly distributed across 10 mutually exclusive classes, such as airplane, automobile, bird, and cat. The dataset contains 50,000 training images and 10,000 test images.

\noindent  \textit{ {2) CIFAR100}}~\cite{KAH+09} shares the same structure and image resolution as CIFAR10, but with a more fine-grained label space comprising 100 distinct classes, grouped into 20 superclasses. Each class contains 600 images, with 500 used for training and 100 for testing. 

\noindent  \textit{ {3) TinyImageNet}}~\cite{YLX+15} is constructed to provide a more computationally feasible alternative to the full ImageNet dataset. It contains 200 classes, each with 500 training images, 50 validation images, and 50 test images, with all images resized to 64×64 pixels. 

We distribute the training set across all clients in either an IID or non-IID manner, ensuring no overlap between clients. The former refers to an even distribution of the training samples among all clients, whereas the latter represents a skewed allocation of the data. Additionally, all clients share the same test dataset.

\section{Detailed Attack Algorithms}
\label{app:attack}

We evaluate seven trajectory-based attacks covering both client-side and server-side settings, including:

\noindent \textit{ {FedMIA}}~\cite{ZGL+25} leverages model updates from non-target clients and employs a one-tailed likelihood-ratio test to exploit the available adversarial information. It first extracts a low-dimensional measurement of a query sample from local models, estimates the OUT distribution for each communication round, and then infers the membership status based on the measurement trajectory.
Depending on the measurement type, FedMIA-I adopts model loss~\cite{YSG+18}, whereas FedMIA-II employs Grad-Cosine similarity~\cite{LJL+23}.

\noindent \textit{ {FTA}}~\cite{CHE+24} is a training-free attack designed to audit membership privacy risks in FL. The method quantifies the slope of changes in model performance metrics, such as prediction confidence and loss, across the communication rounds to differentiate members from non-members. 
We evaluate two variants: FTA-C, which exploits slope trajectories of model confidence, and FTA-L, which exploits those of model loss.

\noindent \textit{ {SeqMIA}}\cite{LHL+24}, a state-of-the-art trajectory-based MIA originally proposed for the centralized setting, is adapted in this work to the federated setting. It infers membership status by aggregating a sequence of informative training indicators, including sample loss\cite{YSG+18}, maximum output posterior~\cite{SAZ+18, SLM+21}, output posterior dispersion~\cite{SAZ+18}, prediction entropy~\cite{SAZ+18, SLM+21}, and modified prediction entropy~\cite{SLM+21}, as inputs to a recurrent neural network (RNN)-based attack model. In our setting, SeqMIA is applied to perform MIAs against local models.

In addition to the above, we evaluate two further MIAs: \textit{Loss-Series}, which tracks the trajectory of average loss~\cite{YSG+18}, and \textit{Avg-Cosine}, which measures the average change in gradient cosine similarity~\cite{LJL+23}.
% We exclude computationally expensive attacks with inferior effectiveness~\cite{NMS+19} relative to trajectory-based approaches.

For these attacks, we extract attack features and corresponding model snapshots every 10 communication rounds to balance attack overhead and effectiveness. 
Each attack is evaluated using 5,000 member samples and 10,000 non-member samples, which consist of 1,000 non-member-OFL samples and 9,000 non-member-IFL samples, with the latter equally contributed by all other clients.

\section{Detailed Defense Baselines}
\label{app:defense}

We evaluate our framework against six representative defense approaches, detailed as follows:

\noindent \textit{ {GradSparse}}: 
A common defense strategy in FL to mitigate privacy leakage involves limiting the exposure of model updates. A representative method is gradient sparsification~\cite{GOR+18}, which zeroes out gradients with smaller magnitudes, preserving only the top components ranked by absolute value.
In our experiments, we evaluate gradient sparsification with sparsity rates of $\{0.1, 0.2\}$ and report the best attack performance observed.

\noindent \textit{ {GradNoise}}:
Gradient noise~\cite{BLH+24, ZGL+25} mitigates privacy risks by adding random Gaussian noise to gradients before transmission to the central server, obscuring the influence of individual data points and hindering inference attacks.
Formally, the perturbed gradient is defined as \(\tilde{g} = g + \mathcal{N}(0, \sigma^2 I)\), where \(g\) is the original gradient and \(\mathcal{N}(0, \sigma^2 I)\) is Gaussian noise with standard deviation \(\sigma\). Higher \(\sigma\) improves privacy but may degrade performance. We experiment with \(\sigma \in \{0.01, 0.1\}\) to assess defense effectiveness across different privacy-utility trade-offs.

\noindent \textit{ {DPSGD}}:
DPSGD~\cite{AMZ+24, DSB+22, STS+23} is a widely used defense against membership inference attacks, offering formal privacy guarantees. Here, we focus on client-level DPSGD, where each client clips its local gradients to a predefined threshold and adds Gaussian noise for privacy protection.
Following~\cite{AMZ+24}, we implement a high-utility DPSGD baseline by replacing batch normalization with group normalization and incorporating augmentation multiplicity~\cite{DSB+23} to improve utility.
To explore the privacy-utility trade-off, we experiment with relatively large privacy budgets, varying the noise multiplier among $\{0.2, 0.1, 0.05\}$.

\noindent \textit{ {Mixup}}:
Mixup~\cite{ZHC+17, GHL+23} is a data augmentation technique that generates synthetic training samples by linearly interpolating pairs of inputs and their labels, thereby enhancing model generalization.
The interpolation coefficient is drawn from a Beta distribution, typically denoted as $\text{Beta}(\alpha, \alpha)$. In our experiments, we set $\alpha$ as $0.01$ to control the interpolation strength for defending against MIAs.

\noindent \textit{ {RelaxLoss}}:
RelaxLoss~\cite{CDY+22} is a notable defense against MIAs in centralized settings that preserves model utility. It adopts a relaxed training objective that intentionally increases the loss of training samples to narrow the gap between training and testing losses, thereby reducing the distinction between member and non-member data.
We adapt RelaxLoss to our federated setting by modifying the local training process accordingly.

\noindent \textit{ {HAMP}}: 
HAMP~\cite{ZCK+24} is a state-of-the-art defense against MIAs that mitigates model overconfidence. During training, it replaces conventional hard labels with high-entropy soft labels and adds an entropy-based regularizer to penalize overconfident predictions. At inference, HAMP enforces uniformly low confidence by substituting class-probability outputs with those produced from random inputs, while preserving the relative ordering of predicted labels.
In our experiments, we apply only the training-phase defense in the FL setting.

\section{Additional Experimental Results}

\begin{figure}[htbp]
	\centering
	\subfigure[Small coalition (Utility loss: 0.01).]{
		\centering
		\includegraphics[width=0.45\textwidth]{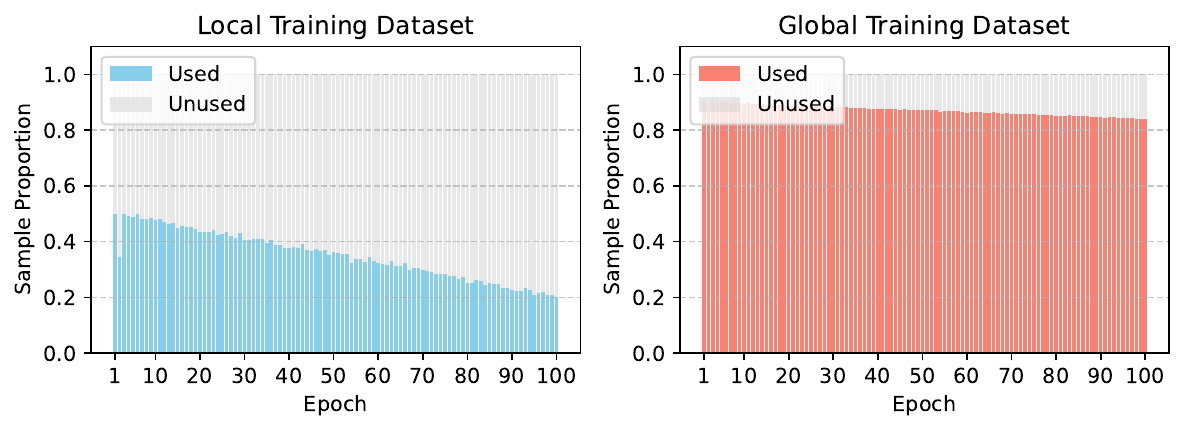}
		\label{fig:U2}
	}\\
	\subfigure[Large coalition (Utility loss: 0.07).]{
		\centering
		\includegraphics[width=0.45\textwidth]{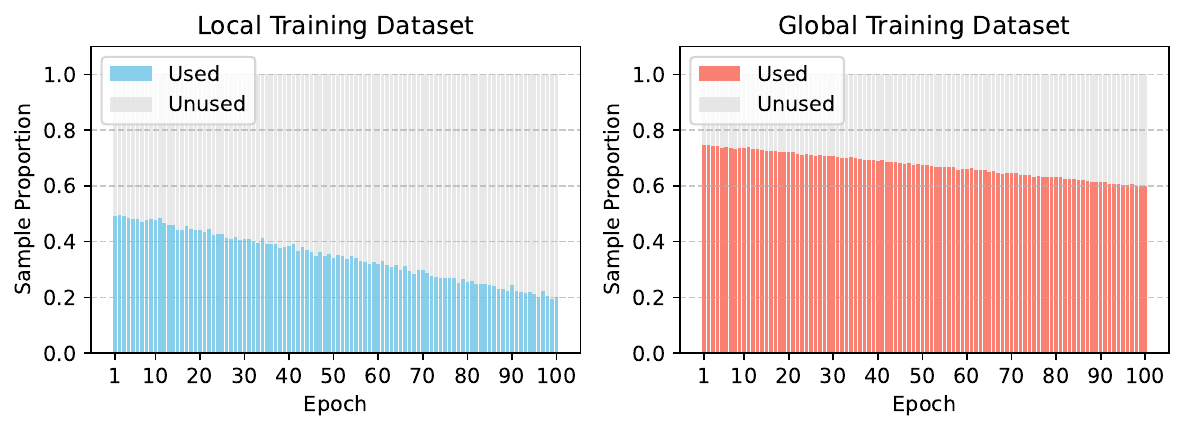}
		\label{fig:U5}
	}
	\caption{%Usage of training samples for a client (decay 50 to 20 classes) and the overall FL system in each round under different coalition sizes on CIFAR10. 
    Usage of training samples per client (decaying from 50 to 20 classes) and across the entire FL system in each round under different coalition sizes on CIFAR10.
    Top: the coalition covers 20\% of the whole training data; Bottom: covers 50\%.
    }
	\label{fig:samplesize}
\end{figure}

\subsection{Training Data Usage}
\label{app:usage}
We investigate the practical usage of training samples from both the perspective of individual clients and the overall FL system, as shown in Figure~\ref{fig:samplesize}
As shown in the left subfigures, our assignment exerts a similar impact on individual clients. In contrast, the right subfigures illustrate that at the coalition level, when more clients participate in defense (i.e., larger coalition size), a greater portion of practical training samples is withheld, resulting in a more pronounced utility loss.

\subsection{More Defense Results}
\label{app:results}
We provide additional comparison results between CoFedMID and the baselines in the Pair and Half settings on CIFAR10 and TinyImageNet using ResNet18, as shown in Tables~\ref{tab:main_res18_c10_AUC}–\ref{tab:main_tiny_res18_TF01}.
Similar to the results on CIFAR100, our framework outperforms the baselines and achieves the lowest average attack performance, with only a small utility drop.
Moreover, we present average results on three datasets using WideResNet-16-3 in Tables~\ref{tab:main_c100_wrn}, \ref{tab:main_c10_wrn}, \ref{tab:main_tiny_wrn}, as the above results have demonstrated the stability of our framework.

\subsection{Micro-level MIA Analysis}
\label{app:microays}
Our defense treats all training samples, classes, and protected clients uniformly. Samples may be from assigned classes or reused as recycled samples, as shown in Figure~\ref{fig:our-data}, and there is no fixed mapping between classes and clients (Algorithm~\ref{alg:partition}). Table~\ref{tab:seqmia_tf01} shows at most 1.27\% TF01 difference across clients and classes.

\begin{table}[htbp]
    \centering
    \caption{TF01 values (\%) of SeqMIA on CIFAR10.}
    \setlength{\tabcolsep}{1.0pt}
    \label{tab:seqmia_tf01}
    \begin{tabular}{c|c|cccccccccc}
    \hline
     & Avg & c0 & c1 & c2 & c3 & c4 & c5 & c6 & c7 & c8 & c9 \\
    \hline
    Client0 & 0.40 & 0.12 & 0.04 & 0.00 & 0.00 & 0.00 & 0.00 & 0.27 & 0.00 & 0.44 & 0.12 \\
    Client4 & 0.27 & 0.00 & 0.08 & 0.46 & 0.10 & 0.39 & 0.01 & 0.81 & 1.27 & 0.06 & 0.16 \\
    \hline
    \end{tabular}
\end{table}

\subsection{Other Ablation Studies}
\label{app:ablation}

\begin{table}[htbp]
    \setlength{\abovecaptionskip}{0pt}  % 表格标题上方间距
    \setlength{\belowcaptionskip}{0pt}  % 表格标题下方间距
    \centering
    \footnotesize
    \caption{Comparison of different decay strategies on CIFAR100 (TF01).}
    \setlength{\tabcolsep}{2.5pt} % 默认是6pt
    \label{tab:decay}
    \begin{tabular}{l|cccc|cc}
    \toprule
     Decay & Loss-Series & FedMIA-II & FTA-C & SeqMIA & Avg & Test acc\\
    \midrule
    Linear      & 2.62 & 0.11 & 3.6  & 0.64 & 1.74 & 0.43 \\
    Cosine      & 4.14 & 0.02 & 2.18 & 1.56 & 1.98 & 0.44 \\
    Exp         & 2.51 & 0.06 & 1.18 & 0.9  & 1.16 & 0.43 \\
    Poly        & 1.68 & 0.01 & 3.94 & 1.01 & 1.66 & 0.43 \\
    \bottomrule
    \end{tabular}
\end{table}

\begin{table}[htbp]
    \setlength{\abovecaptionskip}{0pt}  % 表格标题上方间距
    \setlength{\belowcaptionskip}{0pt}  % 表格标题下方间距
    \centering
    \footnotesize
    \caption{Impact of recycling threshold on CIFAR100 under different coalition sizes (AUC).}
    \label{tab:ab_recycling_ratio}
    \begin{tabular}{lc|cccc}
    \toprule
       \multicolumn{2}{c|}{$r_l$}  &   0.05 & 0.10 & 0.20 &  0.30 \\
    \midrule
    \multirow{2}{*}{Pair}      
        & SeqMIA    & 0.53 & 0.58 & 0.59 & 0.62  \\
        & Test Acc   & 0.47 & 0.46 & 0.46 & 0.47   \\
    \midrule
        \multirow{2}{*}{Half}     
        & SeqMIA    &   0.66 & 0.69 & 0.68 & 0.72  \\
        & Test Acc  &  0.44 & 0.44 & 0.45  &  0.45  \\
    \bottomrule
    \end{tabular}
\end{table}

\textbf{Impact of Different Decay Functions.}
We conduct experiments to evaluate the impact of different decay functions, including linear, cosine, exponential, and polynomial, on defense performance and FL utility.
Table~\ref{tab:decay} reports the comparison of these strategies on CIFAR100 in terms of multiple attack metrics and test accuracy.
The results show that non-linear approaches, such as exponential and polynomial functions, can further mitigate attack performance without incurring additional utility loss.

\textbf{Impact of Recycling Threshold.} % move to Appendix
To prevent excessive recycling of samples, we introduce a recycling threshold that limits the maximum number of samples recycled in each round. We investigate how varying this ratio affects both defense effectiveness and model utility. As shown in Table~\ref{tab:ab_recycling_ratio}, a lower recycling ratio consistently improves defense performance, while having little impact on the FL system’s utility.

\begin{table*}[htbp]
    \setlength{\abovecaptionskip}{0pt}  % 表格标题上方间距
    \setlength{\belowcaptionskip}{0pt}  % 表格标题下方间距
    \centering
    \footnotesize
    \setlength{\tabcolsep}{7pt} % 默认是6pt
    \caption{AUC/TF01 results of defense performance under non-iid settings on CIFAR100.}
    \label{tab:noniid-c100}
    \begin{tabular}{l|c|ccccccc|c}
    \toprule
    \(\beta\) & \textbf{Defense} & \textbf{Loss-Series} & \textbf{Avg-Cosine} & \textbf{FedMIA-I} & \textbf{FedMIA-II} & \textbf{FTA-C} & \textbf{FTA-L} & \textbf{SeqMIA} & \textbf{Test Acc} \\
    \midrule
    \multirow{2}{*}{\(\infty\)}     
        & None      & 0.66/10.82 & 0.80/6.00 & 0.79/4.74 & 0.83/12.64 & 0.74/0.00 & 0.80/4.96 & 0.92/11.80 & 0.46 \\
        & CoFedMID  & 0.46/2.62 & 0.50/1.80 & 0.50/0.50 & 0.57/3.04 & 0.50/0.60 & 0.58/0.56 & 0.50/0.00 & 0.45 \\
    \midrule
    \multirow{2}{*}{10.0}     
        & None      & 0.71/19.32 & 0.78/10.69 & 0.82/19.47 & 0.85/26.51 & 0.68/6.84 & 0.79/1.20 & 0.98/32.93 & 0.47 \\
        & CoFedMID  & 0.48/3.82 & 0.52/2.11 & 0.50/1.47 & 0.59/3.53 & 0.47/0.00 & 0.57/1.36 & 0.54/0.03 & 0.45 \\
    \midrule
    \multirow{2}{*}{1.0}     
        & None      & 0.84/34.52 & 0.89/35.19 & 0.91/52.07 & 0.92/58.57 & 0.56/3.28 & 0.69/0.60 & 0.99/62.01 & 0.45 \\
        & CoFedMID  & 0.59/6.93 & 0.56/3.01 & 0.66/1.79 & 0.64/2.30 & 0.45/0.00 & 0.54/0.16 & 0.52/0.40 & 0.42 \\
    \midrule
    \multirow{2}{*}{0.5}     
        & None      & 0.89/49.91 & 0.94/57.36 & 0.95/75.47 & 0.95/76.01 & 0.49/4.52 & 0.61/0.40 & 1.00/63.77 & 0.38 \\
        & CoFedMID  & 0.70/11.62 & 0.54/3.06 & 0.77/2.18 & 0.66/0.21 & 0.45/3.32 & 0.54/0.12 & 0.58/0.46 & 0.37 \\
    \bottomrule
    \end{tabular}
\end{table*}

\begin{table*}[htbp]
    \setlength{\abovecaptionskip}{0pt}  % 表格标题上方间距
    \setlength{\belowcaptionskip}{0pt}  % 表格标题下方间距
    \centering
    \footnotesize
    \setlength{\tabcolsep}{7pt} % 默认是6pt
    \caption{AUC/TF01 results of defense performance under non-iid settings on CIFAR10.}
    \label{tab:noniid-c10}
    \begin{tabular}{l|c|ccccccc|c}
    \toprule
    \(\beta\) & \textbf{Defense} & \textbf{Loss-Series} & \textbf{Avg-Cosine} & \textbf{FedMIA-I} & \textbf{FedMIA-II} & \textbf{FTA-C} & \textbf{FTA-L} & \textbf{SeqMIA} & \textbf{Test Acc} \\
    \midrule
    \multirow{2}{*}{\(\infty\)}     
        & None      & 0.62/8.36 & 0.59/0.46 & 0.62/0.66 & 0.57/0.54 & 0.56/0.40 & 0.61/0.08 & 0.91/6.89 & 0.77 \\
        & CoFedMID  & 0.49/5.28 & 0.47/8.48 & 0.52/3.58 & 0.38/0.18 & 0.52/0.00 & 0.54/0.00 & 0.61/0.00 & 0.75 \\
    \midrule
    \multirow{2}{*}{10.0}     
        & None      & 0.58/3.85 & 0.61/4.77 & 0.67/4.69 & 0.55/4.73 & 0.57/0.00 & 0.61/0.00 & 0.67/0.35 & 0.77 \\
        & CoFedMID  & 0.50/2.24 & 0.47/5.63 & 0.54/7.15 & 0.46/0.04 & 0.51/1.12 & 0.54/0.84 & 0.48/1.84 & 0.72 \\

    \midrule
    \multirow{2}{*}{1.0}     
        & None      & 0.90/38.14 & 0.90/34.26 & 0.91/49.36 & 0.94/47.18 & 0.56/0.00 & 0.59/0.00 & 0.99/56.09 & 0.73 \\
        & CoFedMID  & 0.70/0.00 & 0.58/13.94 & 0.70/9.68 & 0.73/0.10 & 0.48/0.64 & 0.48/0.72 & 0.53/0.00 & 0.70 \\
    \midrule
    \multirow{2}{*}{0.5}     
        & None      & 0.86/32.73 & 0.85/24.98 & 0.85/43.46 & 0.90/31.16 & 0.58/0.60 & 0.60/0.00 & 0.99/32.49 & 0.68 \\
        & CoFedMID  & 0.65/0.00 & 0.41/1.81 & 0.53/0.00 & 0.65/0.00 & 0.48/0.08 & 0.52/0.56 & 0.51/0.08 & 0.66 \\
    \bottomrule
    \end{tabular}
\end{table*}

\begin{table*}[htpb]
    \setlength{\abovecaptionskip}{0pt}  % 表格标题上方间距
    \setlength{\belowcaptionskip}{0pt}  % 表格标题下方间距
    \centering
    \footnotesize
    \caption{AUC results of defense methods against different attacks on CIFAR10 with ResNet18.
    \textbf{Avg} is the average over seven MIAs, and $\Delta$\textbf{Acc} denotes the test accuracy change relative to the undefended FL (original accuracy).
    % \textbf{Test Acc} denotes the test accuracy of an undefended/defended FL.
    ‘Pair’ denotes a defender coalition of two clients, whereas ‘Half’ refers to a coalition comprising half of the clients in the FL system.
    The best average results are highlighted in \textbf{bold}.
    }
    \label{tab:main_res18_c10_AUC}
    \begin{tabular}{l|l|ccccccc|cc}
    \toprule
    Case & \textbf{Defense} & \textbf{Loss-Series} & \textbf{Avg-Cosine} & \textbf{FedMIA-I} & \textbf{FedMIA-II} & \textbf{FTA-C} & \textbf{FTA-L} & \textbf{SeqMIA} & \textbf{Avg} $\downarrow$ & $\Delta$\textbf{Acc} $\uparrow$\\
    \midrule
    \multirow{2}{*}{} 
        & None  & $0.61 \pm 0.01$ & $0.60 \pm 0.02$ & $0.56 \pm 0.03$ & $0.62 \pm 0.02$ & $0.56 \pm 0.01$ & $0.60 \pm 0.02$ & $0.88 \pm 0.03$ & 0.63 & (0.75) \\ % 0.7587 \\
        \midrule
    \multirow{7}{*}{Pair} 
        & GradSparse  & $0.61 \pm 0.01$ & $0.60 \pm 0.00$ & $0.55 \pm 0.03$ & $0.61 \pm 0.01$ & $0.55 \pm 0.01$ & $0.59 \pm 0.01$ & $0.90 \pm 0.02$ & $0.63$ & +0.02 \\ %0.7734 \\
        & GradNoise   & $0.60 \pm 0.02$ & $0.57 \pm 0.01$ & $0.57 \pm 0.03$ & $0.54 \pm 0.00$ & $0.54 \pm 0.02$ & $0.59 \pm 0.03$ & $0.71 \pm 0.18$ & $0.57$ & -0.06 \\ %0.6936 \\
        & DPSGD  & $0.57 \pm 0.03$ & $0.52 \pm 0.00$ & $0.64 \pm 0.04$ & $0.54 \pm 0.00$ & $0.52 \pm 0.01$ & $0.52 \pm 0.01$ & $0.81 \pm 0.04$ & $0.58$ & -0.03 \\ %{0.7182} \\
        & Mixup  & $0.52 \pm 0.00$ & $0.53 \pm 0.01$ & $0.71 \pm 0.00$ & $0.52 \pm 0.00$ & $0.52 \pm 0.01$ & $0.56 \pm 0.00$ & $0.63 \pm 0.03$ & $0.56$ & -0.00 \\ %{0.7562} \\
        & RelaxLoss   & $0.56 \pm 0.00$ & $0.52 \pm 0.00$ & $0.60 \pm 0.00$ & $0.54 \pm 0.00$ & $0.53 \pm 0.01$ & $0.54 \pm 0.00$ & $0.82 \pm 0.01$ & $0.56$ & -0.01 \\ %{0.7440} \\
        & HAMP   & $0.57 \pm 0.00$ & $0.56 \pm 0.00$ & $0.70 \pm 0.00$ & $0.58 \pm 0.00$ & $0.53 \pm 0.01$ & $0.56 \pm 0.00$ & $0.91 \pm 0.01$ & $0.63$ & +0.01 \\ %{0.7626} \\
        & \textbf{CoFedMID}  & $0.53 \pm 0.01$ & $0.53 \pm 0.01$ & $0.57 \pm 0.01$ & $0.52 \pm 0.01$ & $0.54 \pm 0.01$ & $0.55 \pm 0.01$ & $0.57 \pm 0.04$ & \textbf{0.52} & -0.01 \\ %{0.7454} \\
        \midrule

    \multirow{7}{*}{Half} 
        & GradSparse   & $0.61 \pm 0.01$ & $0.61 \pm 0.00$ & $0.57 \pm 0.03$ & $0.64 \pm 0.00$ & $0.58 \pm 0.01$ & $0.60 \pm 0.01$ & $0.86 \pm 0.02$ & 0.61 & +0.02 \\ %{0.7732} \\
        & GradNoise   & $0.62 \pm 0.00$ & $0.60 \pm 0.00$ & $0.56 \pm 0.00$ & $0.57 \pm 0.00$ & $0.59 \pm 0.01$ & $0.63 \pm 0.01$ & $0.86 \pm 0.01$ & 0.57 & -0.00 \\ % {0.7502} \\
        & DPSGD    & $0.58 \pm 0.02$ & $0.52 \pm 0.00$ & $0.56 \pm 0.03$ & $0.54 \pm 0.01$ & $0.52 \pm 0.01$ & $0.51 \pm 0.01$ & $0.82 \pm 0.02$ & \textbf{0.51} & -0.09 \\ %{0.6666} \\
        & Mixup   & $0.52 \pm 0.00$ & $0.52 \pm 0.00$ & $0.68 \pm 0.00$ & $0.52 \pm 0.00$ & $0.53 \pm 0.00$ & $0.54 \pm 0.00$ & $0.65 \pm 0.04$ & 0.55 & -0.12 \\ %{0.6363} \\
        & RelaxLoss    & $0.57 \pm 0.00$ & $0.51 \pm 0.00$ & $0.58 \pm 0.00$ & $0.54 \pm 0.00$ & $0.52 \pm 0.00$ & $0.54 \pm 0.00$ & $0.84 \pm 0.01$ & 0.56 & -0.07 \\ %{0.6824} \\
        & HAMP   & $0.59 \pm 0.00$ & $0.58 \pm 0.00$ & $0.64 \pm 0.00$ & $0.61 \pm 0.00$ & $0.52 \pm 0.01$ & $0.56 \pm 0.01$ & $0.86 \pm 0.05$ & 0.61 & -0.01 \\ %{0.7510} \\
        & \textbf{CoFedMID}  & $0.54 \pm 0.00$ & $0.52 \pm 0.01$ & $0.52 \pm 0.01$ & $0.54 \pm 0.00$ & $0.52 \pm 0.01$ & $0.56 \pm 0.00$ & $0.60 \pm 0.03$ & 0.54 & -0.03 \\ %{0.7224} \\
    \bottomrule
    \end{tabular}
\end{table*}

\begin{table*}[htpb]
    \setlength{\abovecaptionskip}{0pt}  % 表格标题上方间距
    \setlength{\belowcaptionskip}{0pt}  % 表格标题下方间距
    \centering
    \footnotesize
    \caption{TF01 results of defense methods against different attacks on CIFAR10 with ResNet18.}
    \label{tab:main_res18_c10_TF01}
    \begin{tabular}{l|l|ccccccc|cc}
    \toprule
    Case & \textbf{Defense} & \textbf{Loss-Series} & \textbf{Avg-Cosine} & \textbf{FedMIA-I} & \textbf{FedMIA-II} & \textbf{FTA-C} & \textbf{FTA-L} & \textbf{SeqMIA} & \textbf{Avg} $\downarrow$ & $\Delta$\textbf{Acc} $\uparrow$ \\
    \midrule
    \multirow{1}{*}{} 
         & No Defense   & $8.52 \pm 2.13$ & $0.34 \pm 0.31$ & $0.60 \pm 0.58$ & $1.17 \pm 1.04$ & $0.75 \pm 0.57$ & $0.35 \pm 0.42$ & $7.34 \pm 4.22$ & 2.44 & (0.75) \\
        \midrule
    \multirow{7}{*}{Pair} 
         & GradSparse   & $10.23 \pm 0.69$ & $0.36 \pm 0.20$ & $0.14 \pm 0.12$ & $0.46 \pm 0.09$ & $0.51 \pm 0.28$ & $0.58 \pm 0.30$ & $8.03 \pm 3.13$ & $2.90$ & +0.02 \\
        & GradNoise  & $9.90 \pm 2.63$ &  {$0.26 \pm 0.24$} & $3.74 \pm 3.33$ & $5.97 \pm 4.60$ & $0.20 \pm 0.29$ & $0.34 \pm 0.34$ & $2.91 \pm 3.09$ & $3.33$ & -0.06 \\
        & DPSGD       & $11.82 \pm 1.91$ & $1.23 \pm 0.16$ & $2.85 \pm 1.11$ & $0.32 \pm 0.18$ & $0.15 \pm 0.24$ & $0.33 \pm 0.37$ & $4.20 \pm 3.17$ & $2.99$ & -0.03 \\
        & Mixup       & $10.86 \pm 0.19$ & $5.98 \pm 0.12$ & $8.83 \pm 0.21$ & $11.72 \pm 0.37$ & $2.42 \pm 0.51$ &  $0.02 \pm 0.05$ &  $0.10 \pm 0.17$ & $5.70$ & -0.00 \\
        & RelaxLoss   &  {$6.56 \pm 0.02$} & $0.00 \pm 0.00$ &  $0.00 \pm 0.00$ &  {$0.30 \pm 0.02$} & $0.90 \pm 0.19$ & $0.35 \pm 0.15$ & $3.39 \pm 1.01$ & $1.64$ & -0.01 \\
        & HAMP        & $12.34 \pm 0.08$ & $4.99 \pm 0.18$ & $0.65 \pm 0.01$ & $5.14 \pm 0.51$ &  {$0.06 \pm 0.08$} & $0.30 \pm 0.12$ & $9.45 \pm 5.32$ & $4.70$ & +0.01 \\
        & \textbf{CoFedMID}    &  $1.47 \pm 1.12$ &  $0.00 \pm 0.00$ &  {$0.03 \pm 0.01$} &  $0.07 \pm 0.07$ &  $0.04 \pm 0.06$ &  {$0.12 \pm 0.17$} &  {$0.52 \pm 0.47$} &  \textbf{0.32} & -0.01 \\
        \midrule
    \multirow{7}{*}{Half}
        & GradSparse  &  {$8.46 \pm 0.31$} &  $0.11 \pm 0.04$ &  $0.09 \pm 0.01$ &  $0.10 \pm 0.02$ & $0.48 \pm 0.22$ &  $0.13 \pm 0.17$ & $6.05 \pm 3.13$ & 2.64 & +0.02 \\
        & GradNoise   & $10.95 \pm 0.16$ &  {$0.17 \pm 0.02$} & $0.34 \pm 0.01$ &  {$0.44 \pm 0.06$} &  $0.17 \pm 0.16$ & $0.42 \pm 0.24$ & $5.22 \pm 2.29$ & 1.61 & -0.00 \\
        & DPSGD       & $13.55 \pm 0.92$ & $0.91 \pm 0.31$ & $0.68 \pm 0.24$ & $0.57 \pm 0.09$ & $0.24 \pm 0.25$ & $0.38 \pm 0.42$ & $7.85 \pm 2.56$ & 2.64 & -0.09 \\
        & Mixup       & $9.97 \pm 0.04$ & $4.02 \pm 0.06$ & $3.55 \pm 0.10$ & $6.24 \pm 0.52$ & $3.99 \pm 0.23$ & $1.06 \pm 0.21$ &  $0.72 \pm 0.65$ & 3.22 & -0.12 \\
        & RelaxLoss    & $12.01 \pm 0.12$ & $0.76 \pm 0.00$ &  {$0.19 \pm 0.01$} & $0.90 \pm 0.09$ & $1.63 \pm 0.29$ &  {$0.27 \pm 0.16$} & $3.38 \pm 1.99$ & 1.28 & -0.07 \\
        & HAMP        & $14.09 \pm 0.12$ & $6.86 \pm 0.21$ & $0.75 \pm 0.08$ & $7.70 \pm 0.38$ & $0.54 \pm 0.08$ & $0.58 \pm 0.22$ & $9.89 \pm 4.88$ & 5.49 & -0.01 \\
        & \textbf{CoFedMID}     &  $3.16 \pm 0.98$ & $0.33 \pm 0.26$ & $0.23 \pm 0.34$ & $0.77 \pm 0.74$ &  {$0.22 \pm 0.39$} & $0.44 \pm 0.29$ &  {$1.74 \pm 0.83$} &  \textbf{1.13} & -0.03 \\
    \bottomrule
    \end{tabular}
\end{table*}

\begin{table*}[htpb]
    \setlength{\abovecaptionskip}{0pt}  % 表格标题上方间距
    \setlength{\belowcaptionskip}{0pt}  % 表格标题下方间距
    \centering
    \footnotesize
    \caption{AUC results of defense methods against different attacks on TinyImageNet with ResNet18.}
    \label{tab:main_tiny_res18_AUC}
    \begin{tabular}{l|l|ccccccc|cc}
    \toprule
    Case & \textbf{Defense} & \textbf{Loss-Series} & \textbf{Avg-Cosine} & \textbf{FedMIA-I} & \textbf{FedMIA-II} & \textbf{FTA-C} & \textbf{FTA-L} & \textbf{SeqMIA} & \textbf{Avg} $\downarrow$ & $\Delta$\textbf{Acc} $\uparrow$ \\
    \midrule
    \multirow{2}{*}{} 
        & None  & $0.65 \pm 0.01$ & $0.79 \pm 0.02$ & $0.79 \pm 0.02$ & $0.83 \pm 0.01$ & $0.57 \pm 0.03$ & $0.72 \pm 0.01$ & $0.98 \pm 0.01$ & 0.76 & (0.39) \\
        \midrule
        
    \multirow{7}{*}{Pair} 
        & GradSparse  & $0.64 \pm 0.01$ & $0.78 \pm 0.01$ & $0.76 \pm 0.02$ & $0.82 \pm 0.01$ & $0.54 \pm 0.01$ & $0.70 \pm 0.01$ & $0.97 \pm 0.01$ & 0.74 & +0.02 \\
        & GradNoise   & $0.60 \pm 0.00$ & $0.70 \pm 0.00$ & $0.56 \pm 0.00$ & $0.65 \pm 0.00$ & $0.53 \pm 0.00$ & $0.70 \pm 0.00$ & $0.93 \pm 0.01$ & 0.67 & +0.01 \\
        & DPSGD   & $0.52 \pm 0.00$ & $0.55 \pm 0.00$ & $0.60 \pm 0.00$ & $0.54 \pm 0.00$ & $0.57 \pm 0.00$ & $0.53 \pm 0.01$ & $0.80 \pm 0.01$ & 0.54 & -0.13 \\
        & Mixup  & $0.54 \pm 0.00$ & $0.60 \pm 0.00$ & $0.62 \pm 0.00$ & $0.57 \pm 0.00$ & $0.54 \pm 0.00$ & $0.55 \pm 0.01$ & $0.87 \pm 0.02$ & 0.56 & -0.02 \\
        & RelaxLoss  & $0.63 \pm 0.00$ & $0.70 \pm 0.00$ & $0.73 \pm 0.00$ & $0.79 \pm 0.00$ & $0.56 \pm 0.01$ & $0.72 \pm 0.00$ & $0.96 \pm 0.00$ & 0.73 & +0.01 \\
        & HAMP   & $0.56 \pm 0.00$ & $0.52 \pm 0.00$ & $0.70 \pm 0.00$ & $0.51 \pm 0.00$ & $0.56 \pm 0.00$ & $0.53 \pm 0.01$ & $0.98 \pm 0.00$ & 0.59 & -0.02 \\
        & \textbf{CoFedMID}  & $0.51 \pm 0.01$ & $0.50 \pm 0.01$ & $0.53 \pm 0.00$ & $0.56 \pm 0.01$ & $0.54 \pm 0.01$ & $0.57 \pm 0.01$ & $0.58 \pm 0.01$ & \textbf{0.52} & -0.01 \\
        \midrule
        
    \multirow{7}{*}{Half} 
        & GradSparse & $0.65 \pm 0.00$ & $0.78 \pm 0.01$ & $0.78 \pm 0.01$ & $0.82 \pm 0.00$ & $0.54 \pm 0.01$ & $0.71 \pm 0.01$ & $0.97 \pm 0.00$ & 0.75 & -0.00 \\
        & GradNoise   & $0.60 \pm 0.00$ & $0.71 \pm 0.00$ & $0.66 \pm 0.00$ & $0.67 \pm 0.00$ & $0.50 \pm 0.00$ & $0.68 \pm 0.01$ & $0.95 \pm 0.01$ & 0.68 & -0.02 \\
        & DPSGD    & $0.53 \pm 0.00$ & $0.54 \pm 0.00$ & $0.53 \pm 0.00$ & $0.53 \pm 0.00$ & $0.59 \pm 0.01$ & $0.51 \pm 0.00$ & $0.86 \pm 0.01$ & 0.55 & -0.25 \\
        & Mixup    & $0.53 \pm 0.00$ & $0.57 \pm 0.00$ & $0.56 \pm 0.00$ & $0.55 \pm 0.00$ & $0.58 \pm 0.01$ & $0.52 \pm 0.00$ & $0.92 \pm 0.01$ & 0.56 & -0.09 \\
        & RelaxLoss   & $0.64 \pm 0.00$ & $0.73 \pm 0.00$ & $0.77 \pm 0.00$ & $0.81 \pm 0.00$ & $0.55 \pm 0.01$ & $0.72 \pm 0.00$ & $0.97 \pm 0.00$ & 0.74 & -0.02 \\
        & HAMP   & $0.53 \pm 0.00$ & $0.50 \pm 0.00$ & $0.64 \pm 0.00$ & $0.50 \pm 0.00$ & $0.58 \pm 0.00$ & $0.52 \pm 0.01$ & $0.99 \pm 0.01$ & 0.58 & -0.11 \\
        & \textbf{CoFedMID}   & $0.52 \pm 0.00$ & $0.51 \pm 0.01$ & $0.50 \pm 0.01$ & $0.53 \pm 0.00$ & $0.56 \pm 0.01$ & $0.55 \pm 0.01$ & $0.61 \pm 0.05$ & \textbf{0.52} & -0.02 \\
    \bottomrule
    \end{tabular}
\end{table*}

\begin{table*}[htpb]
    \setlength{\abovecaptionskip}{0pt}  % 表格标题上方间距
    \setlength{\belowcaptionskip}{0pt}  % 表格标题下方间距
    \centering
    \footnotesize
    \caption{TF01 results of defense methods against different attacks on TinyImageNet with ResNet18.}
    \label{tab:main_tiny_res18_TF01}
    \setlength{\tabcolsep}{4.5pt} % 默认是6pt
    \begin{tabular}{l|l|ccccccc|rc}
    \toprule
    Case & \textbf{Defense} & \textbf{Loss-Series} & \textbf{Avg-Cosine} & \textbf{FedMIA-I} & \textbf{FedMIA-II} & \textbf{FTA-C} & \textbf{FTA-L} & \textbf{SeqMIA} & \textbf{Avg} $\downarrow$ & $\Delta$\textbf{Acc} $\uparrow$ \\
    \midrule
    \multirow{2}{*}{} 
        & None & $15.68 \pm 2.10$ & $7.72 \pm 1.64$ & $8.94 \pm 2.05$ & $15.66 \pm 1.65$ & $4.16 \pm 1.20$ & $0.24 \pm 0.14$ & $21.49 \pm 1.93$ & 10.53  & (0.39) \\ % & 0.3935\\
        \midrule
        
    \multirow{7}{*}{Pair} 
        & GradSparse  & $14.35 \pm 0.10$ & $7.61 \pm 0.95$ & $8.96 \pm 0.60$ & $11.75 \pm 1.49$ & $2.49 \pm 1.11$ & $1.11 \pm 0.35$ & $25.24 \pm 7.39$ & 10.22 & +0.02 \\ %0.4144 \\
        & GradNoise  & $16.18 \pm 0.01$ & $3.74 \pm 0.10$ & $3.03 \pm 0.09$ & $4.51 \pm 0.22$ & $3.14 \pm 0.20$ & $1.06 \pm 0.10$ & $8.87 \pm 3.50$ & 5.79 & +0.01 \\ % 0.4094 \\
        & DPSGD & $9.03 \pm 0.34$ & $1.94 \pm 0.05$ & $10.80 \pm 0.16$ & $0.21 \pm 0.03$ & $1.94 \pm 1.13$ & $0.26 \pm 0.10$ & $3.73 \pm 3.22$ & 3.99 & -0.13 \\ %0.2627\\
        & Mixup  & $10.49 \pm 0.10$ & $10.23 \pm 0.50$ & $13.81 \pm 0.27$ & $15.74 \pm 0.39$ & $4.86 \pm 0.56$ & $0.57 \pm 0.06$ & $1.55 \pm 0.93$ & 8.18 & -0.02 \\ %0.3724 \\
        & RelaxLoss & $14.28 \pm 0.01$ & $2.23 \pm 0.04$ & $6.03 \pm 0.02$ & $7.84 \pm 0.23$ & $6.18 \pm 0.38$ & $1.58 \pm 0.22$ & $17.21 \pm 2.65$ & 7.91 & +0.01 \\ %0.4097\\
        & HAMP  & $11.92 \pm 0.00$ & $11.92 \pm 0.17$ & $0.18 \pm 0.04$ & $15.14 \pm 0.11$ & $3.41 \pm 0.28$ & $0.07 \pm 0.03$ & $37.49 \pm 16.02$ & 11.45 & -0.02 \\ %0.3774 \\
        & \textbf{CoFedMID}  & $3.64 \pm 0.23$ & $1.58 \pm 0.36$ & $0.14 \pm 0.10$ & $2.04 \pm 0.79$ & $3.33 \pm 3.16$ & $0.55 \pm 0.33$ & $0.49 \pm 0.39$ & \textbf{1.68} & -0.01\\ % 0.3807 \\
        \midrule
        
    \multirow{7}{*}{Half} 
        & GradSparse  & $14.22 \pm 0.82$ & $7.42 \pm 0.93$ & $7.76 \pm 1.07$ & $13.93 \pm 0.53$ & $4.61 \pm 2.17$ & $0.01 \pm 0.02$ & $22.49 \pm 3.99$ & 10.06 & -0.00 \\ % 0.4076 \\
        & GradNoise  & $15.62 \pm 0.00$ & $2.95 \pm 0.15$ & $2.29 \pm 0.18$ & $4.77 \pm 0.06$ & $6.54 \pm 0.52$ & $0.70 \pm 0.15$ & $20.40 \pm 5.47$ & 7.61 & -0.02 \\ %0.3766 \\
        & DPSGD  & $9.39 \pm 0.18$ & $2.31 \pm 0.06$ & $7.12 \pm 0.17$ & $0.32 \pm 0.03$ & $0.67 \pm 0.69$ & $0.25 \pm 0.21$ & $7.82 \pm 1.24$ & 4.13 & -0.25 \\ %0.1442 \\
        & Mixup  & $11.54 \pm 0.18$ & $5.05 \pm 0.39$ & $7.65 \pm 0.18$ & $11.72 \pm 0.30$ & $3.06 \pm 0.23$ & $0.52 \pm 0.44$ & $7.02 \pm 6.63$ & 6.65 & -0.09 \\ %0.3061 \\
        & RelaxLoss  & $16.18 \pm 0.11$ & $4.12 \pm 0.13$ & $6.28 \pm 0.22$ & $10.81 \pm 0.49$ & $6.92 \pm 0.71$ & $0.09 \pm 0.12$ & $23.28 \pm 5.44$ & 9.67 & -0.02 \\ %0.3758 \\
        & HAMP  & $12.34 \pm 0.02$ & $9.16 \pm 0.30$ & $0.00 \pm 0.00$ & $7.89 \pm 0.49$ & $2.02 \pm 1.14$ & $0.73 \pm 0.16$ & $47.81 \pm 18.30$ & 11.42 & -0.11 \\ %0.2887 \\
        & \textbf{CoFedMID}   & $6.00 \pm 0.93$ & $3.15 \pm 1.62$ & $0.31 \pm 0.14$ & $0.91 \pm 0.24$ & $2.89 \pm 2.88$ & $0.16 \pm 0.14$ & $0.20 \pm 0.30$ & \textbf{1.94} & -0.02 \\ %0.3767 \\
    \bottomrule
    \end{tabular}
\end{table*}

\begin{table*}[htpb]
    \setlength{\abovecaptionskip}{0pt}  % 表格标题上方间距
    \setlength{\belowcaptionskip}{0pt}  % 表格标题下方间距
    \centering
    \footnotesize
    \caption{AUC/TF01 results of defense methods against different attacks on CIFAR100 with WideResNet.}
    \label{tab:main_c100_wrn}
    \setlength{\tabcolsep}{4.5pt} % 默认是6pt
    \begin{tabular}{l|l|ccccccc|cc}
    \toprule
    Case & \textbf{Defense} & \textbf{Loss-Series} & \textbf{Avg-Cosine} & \textbf{FedMIA-I} & \textbf{FedMIA-II} & \textbf{FTA-C} & \textbf{FTA-L} & \textbf{SeqMIA} & \textbf{Avg} $\downarrow$ & $\Delta$\textbf{Acc} $\uparrow$ \\
    \midrule
    \multirow{2}{*}{} 
    & None & 0.70/4.84 & 0.84/32.17 & 0.84/27.16 & 0.88/45.59 & 0.82/6.87 & 0.74/1.18 & 0.98/16.72 & 0.83/19.05 & (0.57) \\
    \midrule
    \multirow{7}{*}{Pair} 
    & GradSparse & 0.70/5.34 & 0.84/32.99 & 0.85/29.36 & 0.88/47.78 & 0.83/13.78 & 0.75/0.93 & 0.99/25.33 & 0.83/22.22 & -0.01 \\
    & GradNoise & 0.62/6.02 & 0.74/5.92 & 0.46/3.84 & 0.70/12.80 & 0.66/6.81 & 0.63/0.66 & 0.61/3.06 & 0.63/5.59 & -0.10 \\
    & DPSGD & 0.58/17.76 & 0.56/4.64 & 0.55/2.64 & 0.54/0.49 & 0.52/0.50 & 0.54/0.02 & 0.74/4.34 & 0.58/4.34 & -0.00 \\
    & Mixup & 0.52/7.08 & 0.74/8.40 & 0.51/10.54 & 0.69/18.69 & 0.59/0.21 & 0.69/0.38 & 0.47/0.39 & 0.60/6.53 & -0.03 \\
    & RelaxLoss & 0.63/9.96 & 0.65/1.26 & 0.59/0.81 & 0.78/8.21 & 0.76/1.54 & 0.64/0.21 & 0.64/3.39 & 0.67/3.63 & -0.00 \\
    & HAMP & 0.60/5.66 & 0.65/16.92 & 0.72/10.70 & 0.68/20.77 & 0.54/0.09 & 0.56/0.00 & 0.60/2.86 & 0.61/8.14 & -0.03 \\
    & \textbf{CoFedMID} & 0.52/2.02 & 0.61/1.51 & 0.52/0.70 & 0.63/1.09 & 0.57/2.04 & 0.55/1.21 & 0.56/0.18 & \textbf{0.57}/\textbf{1.25} & -0.01 \\

    \midrule
    \multirow{7}{*}{Half} 
    & GradSparse & 0.71/6.56 & 0.84/34.38 & 0.84/29.07 & 0.88/47.29 & 0.82/14.73 & 0.75/2.58 & 0.98/24.08 & 0.83/22.67 & -0.01 \\
    & GradNoise & 0.62/6.45 & 0.72/6.19 & 0.57/2.32 & 0.70/12.18 & 0.65/3.85 & 0.62/1.25 & 0.63/3.56 & 0.64/5.11 & -0.16 \\
    & DPSGD & 0.61/18.22 & 0.56/5.02 & 0.59/2.50 & 0.53/0.48 & 0.50/0.02 & 0.52/0.00 & 0.86/15.52 & 0.60/5.97 & -0.01 \\
    & Mixup & 0.55/9.02 & 0.70/5.93 & 0.51/5.07 & 0.69/14.49 & 0.64/2.09 & 0.72/1.65 & 0.52/0.63 & 0.62/5.55 & -0.04 \\
    & RelaxLoss & 0.64/11.06 & 0.65/3.34 & 0.62/0.78 & 0.76/6.47 & 0.68/1.56 & 0.62/0.62 & 0.59/6.20 & 0.65/4.29 & -0.02 \\
    & HAMP & 0.55/6.86 & 0.67/18.59 & 0.76/4.91 & 0.70/23.17 & 0.54/0.13 & 0.57/0.33 & 0.61/6.64 & 0.62/8.66 & -0.09 \\
    & \textbf{CoFedMID} & 0.54/2.60 & 0.60/1.23 & 0.52/0.70 & 0.61/1.22 & 0.55/4.94 & 0.54/2.40 & 0.56/0.28 & \textbf{0.56}/\textbf{1.91} & -0.03 \\
    \bottomrule
    \end{tabular}
\end{table*}

\begin{table*}[htpb]
    \setlength{\abovecaptionskip}{0pt}  % 表格标题上方间距
    \setlength{\belowcaptionskip}{0pt}  % 表格标题下方间距
    \centering
    \footnotesize
    \caption{AUC/TF01 results of defense methods against different attacks on CIFAR10 with WideResNet.}
    \label{tab:main_c10_wrn}
    \setlength{\tabcolsep}{4.5pt} % 默认是6pt
    \begin{tabular}{l|l|ccccccc|cc}
    \toprule
    Case & \textbf{Defense} & \textbf{Loss-Series} & \textbf{Avg-Cosine} & \textbf{FedMIA-I} & \textbf{FedMIA-II} & \textbf{FTA-C} & \textbf{FTA-L} & \textbf{SeqMIA} & \textbf{Avg} $\downarrow$ & $\Delta$\textbf{Acc} $\uparrow$ \\
    \midrule
    & None & 0.61/1.54 & 0.63/15.78 & 0.60/13.72 & 0.67/22.89 & 0.57/0.93 & 0.59/0.09 & 0.94/16.06 & 0.66/10.14 & (0.85) \\
    \midrule
    
    \multirow{7}{*}{Pair} 
    & GradSparse & 0.60/2.46 & 0.61/13.10 & 0.57/11.84 & 0.64/23.42 & 0.57/2.60 & 0.59/0.00 & 0.70/1.03 & 0.61/7.78 & -0.01 \\
    & GradNoise & 0.58/4.42 & 0.62/1.92 & 0.34/4.63 & 0.57/7.63 & 0.54/0.82 & 0.55/0.00 & 0.61/0.42 & 0.54/2.83 & -0.08 \\
    & DPSGD & 0.62/12.10 & 0.52/2.02 & 0.43/0.84 & 0.54/0.84 & 0.50/0.28 & 0.51/0.96 & 0.68/0.83 & 0.54/2.55 & -0.02 \\
    & Mixup & 0.50/3.90 & 0.68/9.48 & 0.61/4.24 & 0.52/11.46 & 0.50/0.40 & 0.56/0.80 & 0.51/0.00 & 0.54/4.33 & -0.03 \\
    & RelaxLoss & 0.55/8.06 & 0.56/1.70 & 0.46/6.60 & 0.60/14.74 & 0.54/1.52 & 0.47/0.16 & 0.52/0.45 & \textbf{0.53}/4.75 & -0.03 \\
    & HAMP & 0.59/3.64 & 0.59/15.10 & 0.52/19.68 & 0.61/24.02 & 0.59/2.48 & 0.63/1.32 & 0.67/20.11 & 0.60/12.34 & -0.02 \\
    & \textbf{CoFedMID} & 0.56/3.29 & 0.56/0.03 & 0.58/0.09 & 0.58/0.15 & 0.59/1.20 & 0.60/3.02 & 0.54/2.14 & 0.56/\textbf{1.42} & -0.01 \\

    \midrule
    \multirow{7}{*}{Half} 
    & GradSparse & 0.62/0.84 & 0.63/13.58 & 0.59/15.52 & 0.66/22.96 & 0.58/0.72 & 0.59/0.00 & 0.90/14.25 & 0.65/9.70 & -0.01 \\
    & GradNoise & 0.58/1.68 & 0.61/0.60 & 0.44/0.82 & 0.56/7.38 & 0.54/0.76 & 0.55/0.04 & 0.63/0.86 & 0.56/1.73 & -0.16 \\
    & DPSGD & 0.60/7.91 & 0.54/1.58 & 0.51/0.31 & 0.54/0.58 & 0.51/0.42 & 0.51/0.02 & 0.73/4.12 & 0.56/2.13 & -0.12 \\
    & Mixup & 0.53/4.30 & 0.64/7.96 & 0.62/3.12 & 0.51/7.56 & 0.58/0.16 & 0.62/0.32 & 0.51/0.51 & 0.57/3.42 & -0.05 \\
    & RelaxLoss & 0.54/14.08 & 0.54/0.62 & 0.57/2.22 & 0.56/5.70 & 0.51/0.00 & 0.52/0.04 & 0.52/0.24 & \textbf{0.53}/3.27 & -0.08 \\
    & HAMP & 0.62/7.30 & 0.61/13.92 & 0.57/16.84 & 0.65/23.98 & 0.61/0.00 & 0.63/2.24 & 0.64/15.29 & 0.62/11.37 & -0.04 \\
    & \textbf{CoFedMID} & 0.57/1.03 & 0.56/0.11 & 0.46/0.10 & 0.60/0.85 & 0.58/2.46 & 0.60/1.92 & 0.64/1.69 & 0.57/\textbf{1.17} & -0.02 \\
    \bottomrule
    \end{tabular}
\end{table*}

\begin{table*}[htpb]
    \setlength{\abovecaptionskip}{0pt}  % 表格标题上方间距
    \setlength{\belowcaptionskip}{0pt}  % 表格标题下方间距
    \centering
    \footnotesize
    \caption{AUC/TF01 results of defense methods against different attacks on TinyImageNet with WideResNet.}
    \label{tab:main_tiny_wrn}
    \setlength{\tabcolsep}{4.5pt} % 默认是6pt
    \begin{tabular}{l|l|ccccccc|cc}
    \toprule
    Case & \textbf{Defense} & \textbf{Loss-Series} & \textbf{Avg-Cosine} & \textbf{FedMIA-I} & \textbf{FedMIA-II} & \textbf{FTA-C} & \textbf{FTA-L} & \textbf{SeqMIA} & \textbf{Avg} $\downarrow$ & $\Delta$\textbf{Acc} $\uparrow$ \\
    \midrule
    & None & 0.69/19.41 & 0.84/20.86 & 0.82/17.34 & 0.86/30.66 & 0.64/0.38 & 0.71/0.25 & 0.98/29.25 & 0.79/16.88 & (0.49) \\
    \midrule

    \multirow{7}{*}{Pair} 
    & GradSparse & 0.69/19.66 & 0.85/22.72 & 0.82/17.70 & 0.87/31.88 & 0.64/0.36 & 0.71/0.72 & 0.98/26.15 & 0.79/17.03 & -0.00 \\
    & GradNoise & 0.60/11.74 & 0.70/4.14 & 0.48/7.55 & 0.67/13.37 & 0.57/0.66 & 0.63/0.04 & 0.80/11.64 & 0.64/7.02 & -0.09 \\
    & DPSGD & 0.54/12.44 & 0.55/3.22 & 0.62/7.16 & 0.54/0.56 & 0.54/1.76 & 0.52/0.00 & 0.85/5.23 & 0.59/4.34 & -0.15 \\
    & Mixup & 0.51/13.34 & 0.62/11.04 & 0.63/12.10 & 0.59/17.70 & 0.51/0.04 & 0.56/0.12 & 0.83/1.39 & 0.63/7.96 & -0.00 \\
    & RelaxLoss & 0.66/17.72 & 0.65/1.18 & 0.70/4.08 & 0.82/12.32 & 0.62/1.68 & 0.70/0.64 & 0.87/7.56 & 0.72/6.45 & -0.01 \\
    & HAMP & 0.45/14.96 & 0.52/15.34 & 0.73/0.24 & 0.50/20.62 & 0.48/0.00 & 0.51/0.00 & 0.99/40.29 & 0.60/13.06 & -0.01 \\
    & \textbf{CoFedMID} & 0.53/8.79 & 0.61/0.87 & 0.53/0.71 & 0.61/0.42 & 0.52/0.23 & 0.60/0.05 & 0.60/0.00 & \textbf{0.57}/\textbf{1.58} & -0.01 \\

    \midrule
    \multirow{7}{*}{Half} 
    & GradSparse & 0.70/19.50 & 0.85/18.12 & 0.82/16.74 & 0.87/29.74 & 0.64/1.20 & 0.73/0.08 & 0.98/39.42 & 0.80/17.83 & -0.09 \\
    & GradNoise & 0.59/12.02 & 0.68/3.50 & 0.56/5.07 & 0.66/11.57 & 0.55/0.50 & 0.61/0.20 & 0.76/7.44 & 0.63/5.76 & -0.17 \\
    & DPSGD & 0.54/13.02 & 0.56/3.86 & 0.51/4.12 & 0.53/0.16 & 0.55/1.24 & 0.55/0.00 & 0.91/9.43 & 0.59/4.55 & -0.17 \\
    & Mixup & 0.51/13.08 & 0.62/6.10 & 0.45/5.56 & 0.61/14.34 & 0.46/3.32 & 0.55/0.04 & 0.89/0.64 & 0.58/6.15 & -0.03 \\
    & RelaxLoss & 0.67/18.08 & 0.65/1.98 & 0.71/4.10 & 0.82/12.22 & 0.63/0.64 & 0.72/0.00 & 0.93/11.86 & 0.73/6.98 & -0.01 \\
    & HAMP & 0.48/16.30 & 0.51/11.40 & 0.67/0.00 & 0.48/13.58 & 0.45/1.72 & 0.51/0.00 & 0.99/43.68 & 0.58/12.38 & -0.07 \\
    & \textbf{CoFedMID} & 0.56/9.9 & 0.59/0.56 & 0.53/1.0 & 0.59/5.30 & 0.5/0.48 & 0.58/0.2 & 0.68/0.99 & \textbf{0.57}/\textbf{2.63} & -0.02 \\
    \bottomrule
    \end{tabular}
\end{table*}

\end{document}